\documentclass[a4paper,journal,twoside]{IEEEtran}
\IEEEoverridecommandlockouts

\newcommand\hmmax{0}
\newcommand\bmmax{0}

\usepackage[ruled,vlined,linesnumbered]{algorithm2e} 
\usepackage{graphicx}
\usepackage{tikz}
\usepackage{pgfplots}
\usepackage{xcolor}
\usepackage{color, colortbl}
\usepackage{amsmath}
\usepackage{bbm}
\usepackage{amsfonts, amssymb}
\usepackage{mathtools}
\usepackage{enumerate}
\usepackage{cite}
\usepackage[nolist]{acronym}
\usepackage{amsmath}

\usepackage{booktabs} 		
\usepackage{diagbox}		
\usepackage{upgreek}
\usepackage{bbm}
\usepackage{Files/bm}
\usepackage{Files/winsnotation}
\usepackage{Files/Symbols_OPL}
\usepackage{Files/LVcolours} 

\usepackage{setspace}	 	
\usepackage{comment}
\usepackage{physics}
\usepackage{enumitem}

\newlist{myitemize}{itemize}{3}
\setlist[myitemize,1]{label=1.,leftmargin=1em}
\setlist[myitemize,2]{label=$\rightarrow$,leftmargin=0.75em}
\setlist[myitemize,3]{label=$\diamond$}

\usepackage{array}
\newcolumntype{C}[1]{>{\centering\arraybackslash}p{#1}}

\makeatletter
\def\endthebibliography{%
  \def\@noitemerr{\@latex@warning{Empty `thebibliography' environment}}%
  \endlist
}
\makeatother
\usepackage{amsthm}
\theoremstyle{definition}
\newtheorem{definition}{Definition}
\newtheorem{remark}{Remark}
\newtheorem{theorem}{Theorem}

\usepackage{algpseudocode,amsmath}

\pgfplotsset{compat=1.17}

\begin{document}

\title{Performance Limits of Fault-Tolerant \\ Quantum Error Correction Schemes}



\author{
Lorenzo~Valentini,~\IEEEmembership{Member,~IEEE,} 
Diego Forlivesi,~\IEEEmembership{Graduate~Student~Member,~IEEE,} 
and~Marco~Chiani,~\IEEEmembership{Fellow,~IEEE} 
\thanks{Lorenzo Valentini and Marco Chiani are with the Department of Electrical, Electronic, and Information Engineering ``Guglielmo Marconi'' and CNIT/WiLab, University of Bologna, V.le Risorgimento 2, 40136 Bologna, Italy. E-mail: \{lorenzo.valentini13, marco.chiani\}@unibo.it. 
Diego Forlivesi was previously with the University of Bologna and is currently with Quantinuum, Terrington House, Cambridge, CB2 1NL, United Kingdom. E-mail: diego.forlivesi@quantinuum.com. \\ 
This work was supported in part by  Quantum Labs, Cisco Systems, USA.
}
}

\maketitle 

\begin{acronym}
\small
\acro{AWGN}{additive white Gaussian noise}
\acro{BCH}{Bose–Chaudhuri–Hocquenghem}
\acro{BC}{bubble clustering}
\acro{CDF}{cumulative distribution function}
\acro{CRC}{cyclic redundancy code}
\acro{FT}{fault-tolerant}
\acro{LDPC}{low-density parity-check}
\acro{LUT}{lookup table}
\acro{ML}{maximum likelihood}
\acro{MWPM}{minimum weight perfect matching}
\acro{QECC}{quantum error correcting code}
\acro{PDF}{probability density function}
\acro{PMF}{probability mass function}
\acro{MPS}{matrix product state}
\acro{WEP}{weight enumerator polynomial}
\acro{WE}{weight enumerator}
\acro{BD}{bounded distance}
\acro{QLDPC}{quantum low density parity check}
\acro{CSS}{Calderbank, Shor, and Steane}
\acro{MST}{minimum spanning tree}
\acro{PruST}{pruned spanning tree}
\acro{RFire}{Rapid-Fire}
\acro{UF}{union-find}
\acro{LEMON}{library for efficient modeling and optimization in networks}
\acro{STM}{spanning tree matching}
\acro{i.i.d.}{independent identically distributed}
\acro{QEC}{quantum error correction}
\acro{BP}{belief propagation}
\acro{FT}{fault-tolerant}
\acro{MW}{minimum weight}
\end{acronym}
\setcounter{page}{1}

\begin{abstract}
Quantum error correction (QEC) is essential for realizing scalable quantum computation. 
However, when evaluating its benefits, most analyses assume idealized components, overlooking the imperfections inherent in realistic fault-tolerant (FT) implementations. 
In this paper, we investigate the performance of QEC schemes taking into account that quantum gates and measurements are themselves error-prone. 
We derive bounds for the failure probability of Shor-style FT-QEC schemes using limited structural information, such as the number of flag qubits and quantum gates.
Our analysis separates and quantifies two key contributors to the failure rate: decoding errors and residual errors arising from circuit-level faults. 
The derived bounds highlight fundamental limitations in Shor-style FT-QEC performance and quantify how circuit imperfections degrade error correction capabilities, under the assumption of depolarizing noise.
\end{abstract}


\section{Introduction}

One of the main challenges for quantum information systems concerns mitigating the effects of noise resulting from the inevitable quantum-environment interactions  \cite{Sho:95,Laf:96,Kni:97,FleShoWin:08,RehShi:21}. 
To preserve quantum states against noise, \ac{QEC} is adopted to encode logical qubits into larger systems of physical qubits, employing ancillary qubits to detect and correct errors without violating the no-cloning theorem.
In this context, \ac{QEC} is therefore crucial for achieving quantum computation, quantum memories, and quantum communication systems~\cite{FowMarMar:12, Ter:15, MurLiKim:16, Bab:19, Val24:QComm, Val24:SurfaceRadiation, Sha25:DistributedQEC}. 

Stabilizer codes represent a significant class of quantum error-correcting codes~\cite{Got:09}.
Within this category, \ac{CSS} codes are of particular importance for practical implementations~\cite{CalSho:96, Ste:96}.
In fact, despite the increased overhead associated with \ac{CSS} codes compared to their non-\ac{CSS} counterparts, their use simplifies the design of fault-tolerant quantum computing procedures.
Indeed, the structure of CSS codes allows for the transversal application of a CNOT gate between two code blocks to be a fault-tolerant operation, effectively performing a logical CNOT between the corresponding encoded qubits.
Furthermore, for self-dual CSS codes (for which the $\M{X}$ and $\M{Z}$ generators coincide), applying a transversal Hadamard gate implements a logical Hadamard on the encoded qubits.
Additionally, if the weight of every stabilizer generator is a multiple of four, the transversal application of the $\M{S}$ gate also constitutes a fault-tolerant gadget for the code \cite{Got:96}.
From a decoding perspective, CSS codes offer a significant advantage.
Since each stabilizer generator consists solely of either $\M{X}$ or $\M{Z}$ Pauli operators, $\M{X}$ and $\M{Z}$ errors can be decoded independently, which greatly reduces the overall complexity of the decoder.

As a particular cases of \ac{CSS} codes, topological codes offer the advantage of requiring only nearest-neighbour interactions between qubits~\cite{BraKit:98}, a crucial feature for architectures with limited qubit interaction capabilities.
Since the generator weights of topological codes are typically low, they can be interpreted as structured \ac{QLDPC}.
On the other hand, unstructured \ac{QLDPC} codes admit the use of long-range interactions between qubits~\cite{Fue21:AsymCodes, BreNikEbe:21a, PanKal22:QLDPC, Lev:QuantumTanner22}.
Over the past two decades, various families of topological quantum codes, such as surface codes and color codes, have been proposed~\cite{DenKitLan:02,FowSteGro:09,Bom06:colorCodes,HorFowDev:12,ValForChi25:CylMob}.
Among these, codes that can be embedded in a planar layout are particularly appealing, as they simplify implementation across a variety of qubit technologies~\cite{AchRajAle:22, KriSebLac:22, ZhaYouYe:22, BluDolEve:23, Ach:24}.
Also, the existence of efficient decoding algorithms is an important aspect since decoding latency could induce more decoherence~\cite{Del:21, Hig:22, ForValChi:24STM, Bro:23, For25:BubbleCluster, Val25:ImpactDecLatency}.
In this context, performance analysis serves as a crucial tool for benchmarking different code families and providing system design guidelines~\cite{Ash00:MacwPartI, Ash00:MacwPartII, ForValChi24:JSAC, ForValChi25:MacW, Ash25:CSSBounds}.

\Ac{QEC} is essential for realizing reliable quantum computation. However, for \ac{QEC} to be practically useful, it must be implemented in a \ac{FT} manner~\cite{Got:96, Sho96:CatState, Ste97:SteaneGadget}. 
Quantum information is intrinsically fragile, being highly susceptible not only to decoherence, but also to operational errors, such as errors induced by quantum gates~\cite{AchRajAle:22}.
As a result, practical deployment of \ac{QEC} in quantum architectures requires not only robust codes but also efficient \ac{FT} circuits to implement the \ac{QEC} operations.

To implement a \ac{QEC} scheme it is required a syndrome extraction circuit.
In general, three main approaches exist for \ac{FT} syndrome extraction.
The first is the use of repeated, or Shor-style measurements, where each stabilizer is measured using ancilla qubits, and the measurement is repeated multiple times to suppress errors and ensure reliability \cite{Sho96:CatState}.
The second approach is the Steane method, a gadget specifically designed for CSS codes.
It involves preparing encoded logical $\ket{0}$ and $\ket{+}$ ancilla states, followed by transversal CNOT gates that propagate $\M{X}$ and $\M{Z}$ errors from the data block onto the ancilla, which are then measured to extract the syndrome \cite{Ste97:SteaneGadget}. 
The third approach is the Knill method, which uses quantum teleportation: the data is teleported through a prepared encoded ancilla, and syndrome information is obtained as part of the teleportation process \cite{Kni04:SyndExt}.
Both the Steane and Knill methods require the fault-tolerant preparation of encoded ancilla states, which can be resource-intensive \cite{CroDiVTer07:reasourceSteane}. 
For this reason, in this work we restrict our attention to Shor-style syndrome extraction schemes based on short ancilla circuits with repeated measurements, while noting that the Steane and Knill approaches are somewhat easier to analyze, since repeated measurements are not necessary; for example, Steane-style syndrome extraction was analyzed in Ref.~\cite{ForValChi25:MacW}.
Even in this setting, however, a naive implementation that directly measures each stabilizer is generally not \ac{FT}.
In particular, errors from two-qubit gates can propagate into higher-weight errors on the data qubits. 
These are known as hook errors~\cite{DenKitLan:02}.
A \ac{FT} approach to syndrome extraction involves the use of flag qubits, auxiliary qubits entangled with the syndrome qubit such that any hook errors also affect the flag qubits \cite{Cha18:FlagErrCorr, Cha18:FlagsDet}.
By analyzing the resulting flag pattern, the location of faults can be inferred, allowing for appropriate correction of the induced data qubit errors.
An alternative fault-tolerant syndrome extraction scheme employs cat states (or GHZ states) \cite{Sho96:CatState}. 
These can be prepared fault-tolerantly by measuring their $\M{X}$-type stabilizer generator using flag-based circuits or by post-selection methods~\cite{ForAma25:FTCSS}.
Each qubit in the cat state then interacts with a corresponding data qubit via a controlled two-qubit gate, with the cat qubit as the control, tailored to the stabilizer being measured.
Afterward, a transversal Hadamard is applied to the cat state, followed by measurement in the computational basis.
By analyzing the parity of the measurement outcomes, the eigenvalue of the stabilizer is determined.
Note that any single error affecting a cat state qubit propagates to at most one data qubit, thereby preserving fault-tolerance.

Thus, besides conventional performance analysis on quantum code, a fundamental problem in the theory of \ac{QEC} is to characterize the performance limits of \ac{FT}-\ac{QEC} schemes under faulty information processing.
This includes deriving upper bounds on achievable logical error rates, which serve as important design guidelines for building \ac{FT} quantum systems.
Recent works have taken important steps toward explicit circuit-level analysis without relying on Monte Carlo simulation. 
In particular, tensor-network-based approaches have been developed to compute quantum weight enumerator polynomials more efficiently, in some cases achieving exponential speedups with respect to standard methods \cite{CaoLac24:LegoTensorEnumerators}. 
Building on this, tensor enumerator frameworks have been introduced to analyze \ac{FT} circuits and noise models by representing both unitary operations and quantum channels as tensor objects, enabling explicit counting of error paths.
While this method enable detailed analysis of small to moderately sized circuits, such as syndrome extraction circuits for rotated distance-five surface codes, their computational cost quickly becomes prohibitive \cite{KulLac25:CircuitTensor}.

This paper seeks to address a gap in the \ac{QEC} literature through an investigation of \ac{FT}-\ac{QEC} schemes. Specifically, we focus on the performance of Shor-style syndrome extraction circuits under the assumption that quantum gates and measurements may fail with some probability. To maintain broad relevance within this class of protocols, we derive performance bounds using only limited information about the circuits, such as the number of flag qubits and quantum gates involved, without relying on detailed assumptions like the exact flag decoding strategy.
Moreover, as typically done in \ac{FT}-\ac{QEC}~\cite{Cha18:FlagsDet, ChaCro19:MagicStateFlag, PaeRei11:GoalyCodeStatePrep, ForAma25:FTCSS}, we assume that errors are \ac{i.i.d.} depolarizing channel.
Although tighter bounds could be obtained for specific circuit implementations, we deliberately choose not to pursue such optimizations in favor of broader applicability.
To be more precise, in the paper we analyze the two error events that could occur in a \ac{FT}-\ac{QEC} scheme: $i)$ the decoding failure of the code; $ii)$ the residual errors, which are unnoticed errors generated by the \ac{FT} circuit itself.
Comparing our novel bounds with conventional \ac{QEC} bounds we show the degradation due to faulty circuits.

The key contributions of the paper can be summarized as:
\begin{itemize}
    \item We derive fundamental performance limits for \ac{QEC} using Shor-style \ac{FT} circuits.
    \item we identify and analyze an intrinsic source of error that is inherent to all \ac{FT}-\ac{QEC} schemes, the residual error;
    \item we verify and validate our analytical findings through Monte Carlo simulations, emulating the whole \ac{FT}-\ac{QEC} processing.
\end{itemize}

This paper is organized as follows. Section~\ref{sec:preliminary} introduces preliminary concepts about \ac{QEC} and its \ac{FT} implementation, along with definitions that will be used through this work. 
In Sections~\ref{sec:Pfaildec} and~\ref{sec:Pres}, we derive analytical performance bounds for the decoding error probability and the residual error probability, respectively. 
Section~\ref{sec:Parameters} focuses on the derivation of the parameters used in these performance bounds for both \ac{FT} schemes and quantum codes. 
Finally, numerical results are presented in Section~\ref{sec:NumRes}.

\section{Preliminaries and Background}
\label{sec:preliminary}

\subsection{Quantum Error Correction: Basics}

A qubit is an element of the two-dimensional Hilbert space $\mathcal{H}^{2}$, with basis $\ket{0}$ and $\ket{1}$ \cite{NieChu:10}. 
The Pauli operators $\M{I}, \M{X}, \M{Z}$, and $\M{Y}$, are defined by  $\M{I}\ket{a}=\ket{a}$, $\M{X}\ket{a}=\ket{a\oplus 1}$, $\M{Z}\ket{a}=(-1)^a\ket{a}$, and $\M{Y}\ket{a}=i(-1)^a\ket{a\oplus 1}$ for $a \in \lbrace0,1\rbrace$. These operators either commute (e.g. $\M{I}\M{X}=\M{X}\M{I}$) or anticommute (e.g., $\M{X}\M{Z}=-\M{Z}\M{X}$) with each other. 
When considering Pauli operators on $n$ qubits one could construct the $\mathcal{G}_n$ Pauli group \cite{Got:09,NieChu:10}.
We indicate with $[[n,k,d]]$ a \ac{QECC} that encodes $k$ information qubits (called logical qubits), into a codeword\footnote{A codeword is any state in the code space.} $\ket{\psi}$ of $n$ qubits (called data or physical or codeword qubits), able to correct all patterns up to $t = \lfloor(d-1)/2 \rfloor$ errors and, usually, some higher-weight error patterns. 
%
Using the stabilizer formalism, we choose $n-k$ independent commuting operators 
$G_i \in \mathcal{G}_n$, called generators, where $\mathcal{G}_n$ denotes the Pauli 
group on $n$ qubits. The subgroup $S \subset \mathcal{G}_n$ 
generated by these operators is called the stabilizer group and forms an abelian 
subgroup of $\mathcal{G}_n$ that does not contain $-I$~\cite{Got:09,NieChu:10}.
The code $\mathcal{C}$ is the set of quantum states (or codewords) $\ket{\psi}$ satisfying 
$\M{G}_i \ket{\psi}=\ket{\psi},\, i=1, 2, \ldots, n-k$. 
The error syndrome is obtained by measuring stabilizer generators using ancilla qubits.

Other useful qubit gates which are used in this paper are the Hadamard gate, the CNOT gate, and the CZ gate. 
A Hadamard gate, described as $\M{H} = (\M{X} + \M{Z}) / \sqrt{2}$, 
implements a basis change between the $\M{Z}$ and $\M{X}$ bases.
A CNOT gate acts on two qubits, one referred to as \emph{control} and the other as \emph{target}, according to $\ket{c}\ket{t} \mapsto \ket{c}\ket{t \oplus c}$.
Similarly, a CZ gate acts on two qubits according to $\ket{a}\ket{b} \mapsto (-1)^{ab}\ket{a}\ket{b}$.

\subsection{Quantum Error Correction: Performance}

Given a $[[n, k, d]]$ quantum code with weight enumerator polynomial $A(z) = \sum_w A_w z^w$~\cite{CaoChuLac:24}, its intrinsic error correction capability guarantees that the codeword error rate\footnote{Note that we intentionally avoid using the term ``logical error rate'' here, as every type of error must be treated as a failure, not just logical errors.} over a depolarizing channel with error probability $p$ is
\begin{align}
\label{eq:PLwithAz}
    p_\mathrm{L} = 1 - \sum_{w=0}^n \frac{A_w}{4^k} p^w (1-p)^{n-w}
\end{align}
considering that some errors are in the stabilizer set.
For example, for the $[[7,1,3]]$ Steane code we have $A(z) = 4 + 84 z^4 + 168^6$.
Moreover, for the $[[13,1,3]]$ surface code we have $A(z) = 4 + 32z^3 + 48z^4 + 96z^5 + 304z^6 + 768z^7 + 1812z^8 + 3456z^9 + 4464z^{10} + 3552z^{11} + 1560z^{12} + 288z^{13}$.

Given a $[[n, k, d]]$ quantum code over a depolarizing channel with error probability $p$, we can upper bound its logical error probability by 
\begin{align}
\label{eq:PLBD}
    p_\mathrm{L} \le \sum_{e=0}^n \binom{n}{e} f(e) \, p^e (1-p)^{n-e}
\end{align}
where $f(e)$ is the probability that the decoder fails the decoding when there are exactly $e$ errors in the data qubits.
As a particular case, for a bounded distance decoder we have that 
\begin{align}
\label{eq:PFGivenE_BD}
    f(e) = 
    \begin{dcases}
        0 & 0 \le e\le t \\
        1 & t < e \le n\,.
    \end{dcases}
\end{align}
Adopting \eqref{eq:PFGivenE_BD} instead of the actual $f(e)$ of the decoder provides a valid upper bound that matches the slope in the low error rate regime, but with a constant performance gap (see Fig~\ref{fig:BDAnalysis}).
To match the slope and close the gap, in the low-error regime, we can use 
\begin{align}
\label{eq:PFGivenE_BDwithBeta}
    f(e) = \begin{dcases}
        0 & 0 \le e\le t \\
        1-\beta_{t+1} & e=t+1\\
        1 & t+1 < e \le n\,.
    \end{dcases}
\end{align}
For example, considering the \ac{MWPM} decoder, we have that the $[[13, 1, 3]]$ surface code has $\beta_2 = 267/351$ and the $[[9, 1, 3]]$ rotated surface code has $\beta_2 = 5/9$~\cite{ForValChi25:MacW, ForValChi24:JSAC}.
As another example, the $[[7, 1, 3]]$ Steane code has $\beta_2 = 2/9$ under \ac{MW} decoder.
Having knowledge of the weight enumerator $A(z)$ we can slightly improve \eqref{eq:PFGivenE_BDwithBeta} by including the intrinsic error correction capability of the stabilizer quantum code as
\begin{align}
\label{eq:PFGivenE_ver3}
    f(e) = \begin{dcases}
        0 & 0 \le e\le t \\
        1-\beta_{t+1} & e=t+1\\
        1-\frac{A_e}{4^k \binom{n}{e}} & t+1 < e \le n
    \end{dcases}
\end{align}
However, as can be seen in Fig.~\ref{fig:BDAnalysis}, the inclusion of $A(z)$ does not significantly affect the tightness of the bound.
For a more detailed discussion on the gap between actual and bounded distance performance, see~\cite{ForValChi25:MacW, ForValChi24:JSAC}.

\begin{figure}[t]
    \centering
    \resizebox{\columnwidth}{!}{ 
%
%
\definecolor{mycolor1}{rgb}{0.92157,0.58039,0.52549}%
\definecolor{mycolor2}{rgb}{0.95294,0.87059,0.54118}%
\definecolor{mycolor3}{rgb}{0.49412,0.49804,0.60392}%
\definecolor{mycolor4}{rgb}{0.82745,0.52941,0.67059}%
\begin{tikzpicture}

\begin{axis}[%
name = plot,
width=4.5in,
height=3.5in,
at={(0in,0in)},
scale only axis,
xmode=log,
xmin=0.001,
xmax=0.2,
xminorticks=true,
xlabel style={font=\color{white!15!black}, font = \Large},
xlabel={Channel Error Probability $p$},
tick label style={black, semithick, font=\Large},
ymode=log,
ymin=1e-05,
ymax=1,
yminorticks=true,
ylabel style={font=\color{white!15!black}, font = \Large},
ylabel={ Logical Qubit Error Probability $p_L$},
axis background/.style={fill=white},
title style={font=\bfseries},
legend style={at={(0.03,0.97)}, anchor=north west, legend cell align=left, align=left, draw=white!15!black}
]
\addplot [color=goodRed, line width=1.5pt, draw=none, mark=o, only marks, mark size = 3pt, mark options={solid, goodRed}]
  table[row sep=crcr]{%
0.001	1.763877e-05\\
0.002	6.68696e-05\\
0.005	0.000423116\\
0.01	0.00171359\\
0.02	0.00637511\\
0.05	0.0388048\\
0.1	    0.128041\\
0.2	    0.350877\\
};
\addlegendentry{$[[13,1,3]]$ - Simulation}

\addplot [color=black, densely dotted, line width=1.5pt]
  table[row sep=crcr]{%
0.001	7.74301398605697e-05\\
0.002	0.000307458155811804\\
0.005	0.00187982467076163\\
0.01	0.00724894367793521\\
0.02	0.026951262568848\\
0.05	0.135423859739781\\
0.1	0.3786550197418\\
0.2	0.766353779097601\\
};
\addlegendentry{$[[13,1,3]]$ - Upper Bound \eqref{eq:PLBD} using \eqref{eq:PFGivenE_BD}}

\addplot [color=goodOrange, line width=1.5pt]
  table[row sep=crcr]{%
0.001	1.87989692222347e-05\\
0.002	7.55029411617094e-05\\
0.005	0.00047732718347811\\
0.01	0.00194137850668953\\
0.02	0.0079643207798688\\
0.05	0.0511276860644095\\
0.1	    0.192628098379648\\
0.2	    0.562669250052097\\
};
\addlegendentry{$[[13,1,3]]$ - Upper Bound \eqref{eq:PLBD} using \eqref{eq:PFGivenE_BDwithBeta}}

\addplot [color=goodBlue, dashed, line width=1.5pt]
  table[row sep=crcr]{%
0.001	1.87989612820191e-05\\
0.002	7.550281506794e-05\\
0.005	0.000477322367636324\\
0.01	0.00194130429491231\\
0.02	0.00796321948915221\\
0.05	0.051093373643265\\
0.1	    0.192249484028548\\
0.2	    0.559410664144896\\
};
\addlegendentry{$[[13,1,3]]$ - Upper Bound \eqref{eq:PLBD} using \eqref{eq:PFGivenE_ver3}}

\end{axis}
\end{tikzpicture}%
    }
    \caption{Comparison between upper bounds on the logical error probability.}
    \label{fig:BDAnalysis}
\end{figure}
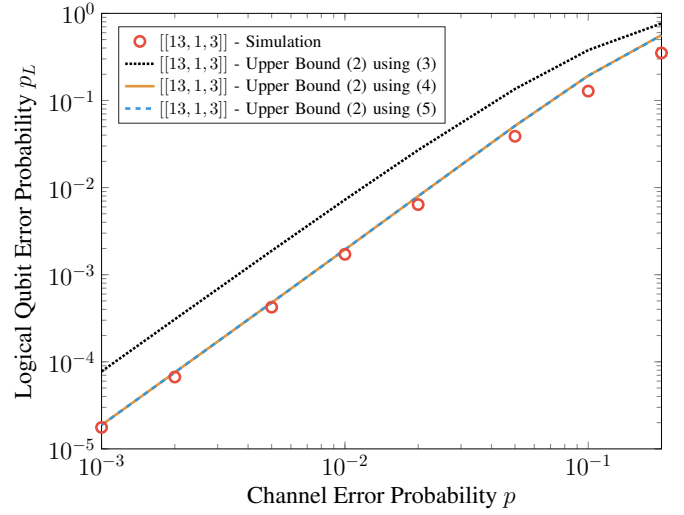

\subsection{Quantum Error Correction: Fault-Tolerance}

\begin{definition}[Faulty location]
    A faulty location in a \ac{FT} quantum circuit is an instance of a basic circuit element, such as a quantum gate or measurement, at which a fault modeled as a Pauli error may occur due to imperfections in physical implementation.
\end{definition}

\begin{definition}[\ac{FT} gadget condition]
    A gadget is said to be $t$-\ac{FT} if any set of up to $e \le t$ Pauli errors occurring at faulty locations within the gadget results in an output with at most $e_\mathrm{f} \le e$ Pauli errors.
\end{definition}

During standard syndrome extraction, a fault arising from one of the two-qubit gates acting on the data qubits can propagate into a higher-weight data qubit error. To ensure the gadget is \ac{FT}, two main alternatives can be considered. One of these is the flag-based syndrome extraction scheme, which introduces additional flag qubits. These qubits are entangled and disentangled during the interaction between the syndrome qubit and the data qubits.
This configuration ensures that any hook error of weight $\leq t$ is also propagated to the flag qubits, making it detectable and correctable. 
For example, in Fig.~\ref{fig:cat_flag}a a flag-based scheme using the approach of \cite{Cha18:FlagErrCorr} is reported. 
Here, an $\M{X}$ error on qubit $a_1$ in position B propagates to the data qubits $q_a$ and $q_b$. However, the same error also propagates to the flag qubit $f_1$ and is detected through its measurement. Based on the flag measurement outcome, the error can be identified and corrected accordingly.
Note that an error is propagated to a flag qubit only if it occurs after the flag qubit is entangled and before it is disentangled from the syndrome qubit. Flag gadgets are specifically designed so that any error occurring outside this window is not harmful. For instance, an error occurring on qubit $a_1$ after the initial Hadamard gate (position A) propagates to all measured data qubits, effectively implementing a stabilizer of the code and therefore leaving the encoded state unaffected. Similarly, an error on qubit $a_1$ in position C propagates only to data qubit $q_a$. In this case, a weight-one error leads to another weight-one error, which is not considered problematic.

\begin{figure}[t]
    \centering
    \resizebox{\columnwidth}{!}{ 
        \input{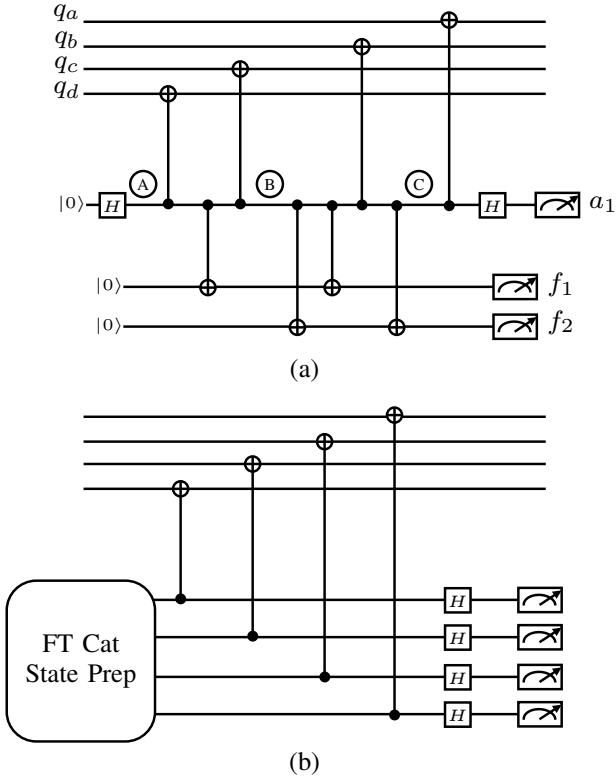}
    }
    \caption{Examples of FT syndrome extraction of a generator with weight four: (a) Flag-based syndrome extraction scheme; (b) Cat-based syndrome extraction scheme.}
    \label{fig:cat_flag}
\end{figure}

Another \ac{FT} approach to syndrome extraction consists of using a scheme based on a cat state of the form $\ket{00\dots0} + \ket{11\dots1}$. 
In this scheme, each data qubit interacts with a different syndrome qubit, as illustrated in Fig.~\ref{fig:cat_flag}b, where a cat-based scheme using the methodology of \cite{Sho96:CatState} is employed. 
As a result, an error occurring on a cat qubit propagates to at most a single data qubit, thereby preserving the \ac{FT} property of the gadget.
An essential component of this construction is the \ac{FT} preparation of the cat state. 
During this stage, it is crucial to ensure that any faults do not propagate into higher-weight errors. 
One method to achieve this is to prepare the cat state by measuring its $\M{X}$ stabilizer generator using flag-based schemes. Alternatively, one can employ preparation gadgets based on post-selection.

Let us define as $\mathcal{F} = \{ \ket{\psi}_\mathrm{out} \neq e^{\imath \theta}\ket{\psi}_\mathrm{in} \,,~ \forall\theta \}$ the failure event in which the encoded state $\ket{\psi}_\mathrm{in}$ is not preserved at the end of the decoding phase.
We denote as $P_\mathrm{fail}$ the probability $\Prob{\mathcal{F}}$.

\begin{definition}[Residual Error]
    We define a residual error as an error that is not part of the stabilizer group of the code, that occurs in the data qubits or is propagated to them, and cannot be detected by the current syndrome extraction procedure.
\end{definition}

As an example, considering a single round of syndrome extraction, a residual error could be generated by the last controlled gate on data qubit of the last syndrome bit extraction.
This failure is endemic to the syndrome extraction whenever quantum gates introduce errors.
We define as $\mathcal{R}$ the event in which a residual error is present.
Then, we can bound $P_\mathrm{fail}$ as
\begin{align}
\label{eq:Pfail_Def}
     P_\mathrm{fail} &= \Prob{\mathcal{F} \mid \mathcal{R}}\Prob{\mathcal{R}} + \Prob{\mathcal{F} \mid \mathcal{R}^c} \Prob{\mathcal{R}^c} \notag \\
     &\le \Prob{\mathcal{R}} + \Prob{\mathcal{F} \mid \mathcal{R}^c} \Prob{\mathcal{R}^c}\,.
\end{align}
Note that, the bound in \eqref{eq:Pfail_Def} is tight due to the fact that the event $\mathcal{F} \mid \mathcal{R}$ almost always occurs. 
The only scenario in which it does not occur is when the decoder fails and the residual error is the same logical operator resulting from the erroneous correction of the decoder, which is unlikely.

\section{Decoder Error Correction Analysis}
\label{sec:Pfaildec}
We denote as $P_\mathrm{fail, dec}$ the probability $\Prob{\mathcal{F} \mid \mathcal{R}^c}$ representing the failure event due to unsuccessful decoding phase after a cycle of \ac{FT} error correction, without considering residual errors.
In particular, we define as \ac{FT} error correction cycle a series of $\rv{R}$ rounds of \ac{FT} syndrome extraction, where $\rv{R}$ is a random variable representing the number of syndrome extractions required to satisfy the policy: observing the same syndrome $t+1$ consecutive times~\cite{Sho96:FT,Bro23:FTConsecutive}.
In the following, we aim at deriving an upper bound for $P_\mathrm{fail, dec}$.

We consider that two primary sources of error are present: $i)$ the channel, which introduces errors before \ac{FT} decoding; $ii)$ the error faults that occur during syndrome extraction at faulty locations.
Specifically, we assume that channel errors are \ac{i.i.d.} with a qubit error probability $p$, while errors in gates and measurements during \ac{FT} operations are also \ac{i.i.d.}, with an error probability denoted by $p_\mathrm{FT}$.
Both of these error sources are modeled as depolarizing channels.
Let us define as $\rv{C}$ the random variable representing the number of errors due to the channel, with \ac{PMF} 
\begin{align}
    p_\rv{C}(c) = \binom{n}{c} p^c (1-p)^{n-c}
\end{align}
Also, let $\rv{S}$ denote the random variable representing the total number of faults that occur throughout the entire $t$-\ac{FT} syndrome extraction circuit.
Then, it can be shown that (see Appendix~\ref{app:PSofs})
\begin{align}
\label{eq:PSofs}
    p_\rv{S}(s) = \sum_{\ell=\lceil s/N_\mathrm{FL}\rceil}^{s} & \sum_{j=0}^{\ell} (-1)^j \binom{\ell}{j} \binom{(\ell-j)N_{\mathrm{FL}}}{s} \notag \\ 
    &\times \frac{
p_{\mathrm{FT}}^{\,s}(1-p_{\mathrm{FT}})^{\ell N_{\mathrm{FL}}-s}
}{
\bigl[1-(1-p_{\mathrm{FT}})^{N_{\mathrm{FL}}}\bigr]^\ell
} P(\ell;t)
\end{align}
where $N_\mathrm{FL}$ is the number of faulty locations in which an error can occur\footnote{Note that $N_\mathrm{FL}$ depends on the code distance $d$. 
However, for the sake of presentation, we refer here to $N_\mathrm{FL}(d)$ as $N_\mathrm{FL}$.}, and $P(\ell; t)$ denotes the probability that exactly $\ell$ syndrome extractions, containing at least one error, occur before the first sequence of $t+1$ consecutive error-free extractions.
In particular, it can be shown that $P(\ell; t)$ follows a geometric distribution (see Appendix~\ref{app:failsBefkSucc})
\begin{align}
\label{eq:Pellt}
    P(\ell; t) &= (1 - p_\mathrm{FT})^{(t+1)N_\mathrm{FL}} (1-(1 - p_\mathrm{FT})^{(t+1)N_\mathrm{FL}})^\ell
\end{align}
where $(1 - p_\mathrm{FT})^{N_\mathrm{FL}}$ is the error-free extraction probability.

As a preliminary upper bound on the logical error probability, we consider the performance of a bounded distance decoder.
This decoder assumes that all error patterns of weight up to the error correction capability $t$ are successfully corrected, while those of weight greater than $t$ are not. 
As such, it provides a conservative estimate due to the fact that any optimal decoder will perform at least as well, and potentially better.
In classical error correction, the bounded distance decoder provides an accurate characterization of optimal performance only for perfect codes. 
For all other codes, optimal decoders can successfully correct some error patterns of weight greater than $t$. 
In quantum error correction, the discrepancy between the bounded distance decoder and the optimal decoder is even more significant due to degeneracy.
In fact, errors that are elements of the stabilizer group are always corrected.
Despite this limitation, the bounded distance bound remains a useful benchmark for evaluating code performance. 
It captures the correct asymptotic behaviour (e.g., slope when $p \ll 1$) of optimal decoders, offering a conservative yet meaningful point of comparison. 
As an example, we provide in Fig.~\ref{fig:BDAnalysis} the comparison between the \ac{BD} decoder and the minimum weight decoder for the $[[13, 1, 3]]$ surface code.

Then, if we search for the probability that at most $t$ errors have occurred into the data qubits we can use the \ac{FT} gadget condition as worst case condition, saying that if $s \le t$ errors have occurred in the syndrome extraction circuit, exactly $s$ errors have been transferred to the data qubits.
This allows us to bound $P_\mathrm{fail, dec}$ as
\begin{align}
\label{eq:Pfail1}
    P_\mathrm{fail, dec} 
    &\le 1 - \sum_{e=0}^t \sum_{m = 0}^e p_\rv{C}(e-m) \, p_\rv{S}(m)\,.
\end{align}

The expression in \eqref{eq:Pfail1} represents an upper bound since some non-failure events are counted as failures. Let us list some of the causes that make this in an upper bound and not an equality.
Firstly, we are considering the bounded distance decoder instead of the adopted good decoder\footnote{Here, we define as good decoder a decoder that ensures the error-correcting capability of the code.}. 
Secondly, we pessimistically assume that $c \le t$ channel errors and $s \le t$ extraction errors always result in $e = c + s$ data qubit errors. However, the actual number of errors could be $e < c + s$ if multiple errors affect the same qubit.
Thirdly, we are not considering that the error correction cycle could finish before encountering $t+1$ non-erroneous rounds of syndrome extraction (e.g., when we have $t+1$ measurement errors on the same syndrome bit).
Fourthly, we have that some faulty locations do not necessarily lead to an error in the data qubit.
In the following, we modify \eqref{eq:Pfail1} to get a tighter upper bound. 

Let us define $\gamma_i$ as the weight of the $i$-th generator, and let $n_\mathrm{flag}(\gamma_i, d)$ denote the number of flag qubits required to construct a fault-tolerant circuit with $\gamma_i$ controls, achieving distance $d$.

\begin{theorem}
\label{th:Pfail2}
We can upper bound the decoding failure occurring after a cycle of \ac{FT} error correction using flag-based syndrome extraction as
\begin{align}
\label{eq:Pfail2}
    &P_\mathrm{fail, dec} \le 1 - \sum_{s=0}^{\tau} p_\rv{S}(s) + (\tau-t) \,p_\rv{S}(t+1) \,P_\mathrm{meas} \notag \\
    &+ \sum_{c=0}^n\sum_{s=0}^{\tau}\sum_{q=0}^s\sum_{e = c}^{\theta} f(e) \,  g (e| c, q) \, p_\rv{C}(c) \, p_\rv{S}(s) p_\rv{Q|S}(q|s)
\end{align}
where
\begin{align}
\label{eq:geGivencs}
    g (e | c, q) &= \binom{n-c}{e-c}\sum_{v=0}^{e-c} (-1)^v \binom{e-c}{v} \left(\frac{e-v}{n}\right)^q \\
    p_\rv{Q|S}(q|s) &= \binom{s}{q} P_\mathrm{fd}^q (1-P_\mathrm{fd})^{s-q}\\
    P_\mathrm{fd} &= \frac{1}{N_\mathrm{FL}} \sum_{i=1}^{n-k}\frac{5\gamma_i-2}{3} + 4 \, n_\mathrm{flag}(\gamma_i, d) \\
    P_\mathrm{meas} &= \sum_{i=1}^{n-k}\left(\frac{2\gamma_i/3+8/3+ 2\,n_\mathrm{flag}(\gamma_i, d)}{N_\mathrm{FL}}\right)^{t+1}\\
    \tau &=\begin{dcases}
        t+1 & \max_i \{\gamma_i/2\} \le t+1\\
        t & \text{otherwise}
    \end{dcases} \\
    \theta &= \min\{n, c+q\}
\end{align}

\end{theorem}

\begin{proof}
    See Appendix~\ref{app:Pfail2}.
\end{proof}

\begin{remark}
    Note that, for all surface codes, rotated surface codes, M\"obius codes, and honeycomb color codes, the condition $\max_i \{\gamma_i/2\} \le t+1$ is always verified.
    For the square-octagon color codes with $d\neq5$ the condition $\max_i \{\gamma_i/2\} \le t+1$ is also verified.
\end{remark}

The key assumptions made in the derivation can be summarized as follows. 
We assume faults independently propagate to data qubits with probability $P_\mathrm{fd}$, yielding a binomial model for the number of propagated errors $\rv{Q}$. 
The total error count combines channel errors and propagated faults and is approximated via an idealized occupancy model $g(e \mid c,q)$, neglecting error cancellations and thus providing a conservative overestimate. 
In estimating $P_\mathrm{fd}$, we ignore recovery by the fault-tolerant flag-based decoder to remain implementation-agnostic. 
Also, we bound the \ac{QEC} decoder performance using $f(e)$. 
Finally, $P_\mathrm{meas}$ is overestimated by neglecting leakage of errors into data qubits.
Due to these assumptions, the derived bound is expected to be tightest at low $p_{FT}$, where fault events are sparse and largely uncorrelated.

\section{Residual Error Analysis}
\label{sec:Pres}

In this section, we derive upper and lower bounds for $\Prob{\mathcal{R}}$, denoted in the following as $P_\mathrm{res}$.
In particular, we derive a lower bound, valid for both flag-based and cat-based syndrome extraction.
Then, for each syndrome extraction scheme we derive an upper bound for $p_\mathrm{FT} \ll 1$.

We begin by defining some quantities that will be used in the subsequent sections.
Let us denote as $v_j^{(\mathrm{z})}$ ($v_j^{(\mathrm{x})}$) the number of $\mathbf{Z}$ ($\mathbf{X}$) operators from the generators acting on the $j$-th qubit.
In other words, $v_j^{(\mathrm{z})}$ ($v_j^{(\mathrm{x})}$) is the weight of the $j$-th column in the $\mathbf{Z}$-part ($\mathbf{X}$-part) of the symplectic representation of the generator matrix.
We denote as $\mathcal{G}_\mathrm{M}$ the larger of the two sets: the indices of all $\mathbf{Z}$-type generators, $\mathcal{G}_z$, and the indices of all $\mathbf{X}$-type generators, $\mathcal{G}_x$.

\subsection{Lower Bound}
To lower bound $P_\mathrm{res}$ we can initially consider only the errors occurring on the last controlled gate on the data qubit.
Then, considering that each qubit has only one faulty location capable of introducing an error, we have
\begin{align}
\label{eq:LB_PRGivenE1}
    P_\mathrm{res} \ge 1 - \sum_{w=0}^n \frac{A_w}{4^k} \, p_\mathrm{FT}^w \, (1-p_\mathrm{FT})^{n-w}\,
\end{align}
where $k$ is the number of information qubits, $A_w$ is the $w$-th coefficient of $A(z)$, and $p_\mathrm{FT}$ is the probability that an error occurs at a faulty location.
As in \eqref{eq:PLwithAz}, here $A(z)$ is used to take into account that residual errors could generate an element of the stabilizer group.

We can refine the bound in \eqref{eq:LB_PRGivenE1} by adding others residual error sources.
In order to do this, we have to consider also the controlled-gate before the last ones.
Without knowledge of the generator order, the best-case scenario arises when CZ and CNOT gates alternate on each qubit. 
This is because, under this arrangement, any error occurring during the third-to-last controlled operation will be detected by at least one of the subsequent two generators acting on the same qubit. 
Note that this reasoning assumes a single fault, which is a conservative assumption and therefore does not affect the validity of the lower bound.
In case of any other ordering we can provide a tighter bound.

\begin{theorem}
\label{th:PresLB}
We can lower bound $P_\mathrm{res}$ as
\begin{align}
\label{eq:PresLB}
    P_\mathrm{res} \ge 1 - \sum_{w=0}^n \frac{A_w}{4^k} \, q^w(p_\mathrm{FT}) \, (1-q(p_\mathrm{FT}))^{n-w}\,
\end{align}
where
\begin{align}
    q(p_\mathrm{FT}) &= \frac{1}{2} + \frac{p_\mathrm{FT}}{3} - \frac{(3-4p_\mathrm{FT})^{v_\mathrm{m}+1}}{6\,(3-2p_\mathrm{FT})^{v_\mathrm{m}}}
\end{align}
and $v_\mathrm{m}$ is the minimum, over qubits, number of consecutive identical controlled gates at the end of the syndrome extraction.
\end{theorem}

\begin{proof}
    See Appendix~\ref{app:PresLB}.
\end{proof}

Note that, if we wish to remain agnostic to the ordering of generator measurements, we conservatively take the best-case value $v_\mathrm{m} = 1$, as we are deriving a lower bound. For example, if all $\mathbf{X}$ generators are measured first, followed by all $\mathbf{Z}$ generators, then $v_\mathrm{m} = \min_{j=1,\dots,n} \{ v_j^{(\mathrm{z})} \}$.

\subsection{Upper Bound}

Deriving an upper bound is nontrivial, as any faulty location in the syndrome extraction circuit can potentially be the endpoint of an error path, resulting in a residual error. 
This implies that one could always construct an error pattern that maps to each fault location, leading to a trivial upper bound that simply counts all faulty locations. 
To avoid this, we focus on deriving an upper bound in the low-error regime ($p_\mathrm{FT} \ll 1$). 
In this regime, it suffices to consider only single-error patterns, which makes the analysis tractable and provides a non-trivial bound.

\begin{theorem}
\label{th:PresUB_Flag}
We can bound $P_\mathrm{res}$, for $p_\mathrm{FT} \ll 1$, as 
\begin{align}
\label{eq:PresUB_Flag}
    P_\mathrm{res} &\le n\,p_\mathrm{FT} + \left(D_\mathrm{res}+A_\mathrm{res}+\frac{3}{2}F_\mathrm{res}\right) \frac{p_\mathrm{FT}}{3}
\end{align}
where $D_\mathrm{res}$ is the number of faulty locations that can cause a residual error from the data qubits (excluding the last gate of each qubit), $A_\mathrm{res}$ is the number of faulty locations that can cause a residual error from the ancillary qubit, and $F_\mathrm{res}$ is the number of faulty locations that can cause a residual error from the flag qubits.
Specifically, we have that
\begin{equation}
\label{eq:PresUB_Flag_Param}
\begin{aligned}
    D_\mathrm{res} &= \max\left\{\sum_{j=0}^{n} v_j^{(\mathrm{z})}, \sum_{j=0}^{n} v_j^{(\mathrm{x})}\right\}\\
    A_\mathrm{res} &= \sum_{i\in \mathcal{G}_\mathrm{M}} \big[K\gamma_i-1+2n_\mathrm{flag}(\gamma_i, d)\big] \\
    F_\mathrm{res} &= 4\sum_{i\in \mathcal{G}_\mathrm{M}} n_\mathrm{flag}(\gamma_i, d)
\end{aligned}
\end{equation}
where $K=1$ for flag-based syndrome extraction and $K=2$ for cat-based syndrome extraction in which the cat is constructed using a $t$-\ac{FT} generator measurement.
\end{theorem}

\begin{proof}
    See Appendix~\ref{app:PresUB_Flag}.
\end{proof}

Note that \eqref{eq:PresUB_Flag} provides an upper bound and not an equality for $p_\mathrm{FT} \ll 1$, since it neglects the possibility that the flag-based correction logic may successfully correct some of the considered error patterns.


\section{Derivation of Analytical Parameters}
\label{sec:Parameters}
In this section, we derive the parameters needed in our analysis for some quantum codes and \ac{FT} syndrome extraction schemes. 
In particular we explicitly provide: $i)$ the total number of faulty locations for \ac{FT} syndrome extraction, $N_\mathrm{FL}(d)$; $ii)$ the number of flag qubits required to make an ancilla qubit with $\gamma_i$ controls \ac{FT} up to distance $d$, $n_\mathrm{flag}(\gamma_i, d)$; $iii)$ the generator weights, $\gamma_i$; $iv)$ the number of $\mathbf{Z}$ ($\mathbf{X}$) operators in the generators that act on the $j$-th qubit, $v_j^{(\mathrm{z})}$ ($v_j^{(\mathrm{x})}$).

In the following, we assume that syndrome extraction for $\mathbf{Z}$ generators is performed with CNOTs, while syndrome extraction for $\mathbf{X}$ generators is performed with CZs. 
In this way, each flag is composed by two CNOT and two measurements (i.e., one for initialization and one to read the flag).

\subsection{Faulty Locations in Flag-based Syndrome Extraction}
Let us define $n_{\mathrm{FL},i}$ as the number of faulty locations in the gadget used to extract the $i$-th syndrome. 
In this way, we have that 
$N_\mathrm{FL}(d) = \sum_{i=1}^{n-k} n_{\mathrm{FL},i}(d)$. 

For a flag-based syndrome extraction, the total number of faulty locations involved in measuring the $i$-th stabilizer generator is given by
\begin{align}
\label{eq:nFPi}
    n_{\mathrm{FL},i}(d) = 2\gamma_i + 4 + 6\, n_\mathrm{flag}(\gamma_i, d)\,.
\end{align}
In \eqref{eq:nFPi}, the term $2\gamma_i$ accounts for the two-qubit gates (e.g., CNOTs or CZs) used to entangle the ancilla with the data qubits according to the $i$-th generator. 
The term $4$ represents the measurement needed both for the initialization of the zero qubit and the final measurement of the ancilla qubit, and the two Hadamard gates. 
The term $6\, n_\mathrm{flag}(\gamma_i, d)$ accounts for the six faulty locations (including both two-qubit gates and measurements) per flag qubit introduced to ensure fault-tolerance in the syndrome extraction.

In the following sections we provide some examples for the value $n_\mathrm{flag}$.

\subsubsection{Flag-based scheme for low-weight generators}
By exhaustive search, we can find the minimum amount of flag required to protect a generator extraction for low $\gamma_i$.
From this analysis, we found that
\begin{align}
\label{eq:nflagOpt}
    n_\mathrm{flag}(\gamma_i, d) &= 
    \begin{dcases}
        0 & \gamma_i \le 3\\
        2 & 3 < \gamma_i \le 5\,.
    \end{dcases}
\end{align}

\subsubsection{Flag-based Chao and Reichardt scheme~\cite{ChaRei20:FlagMethodAnyD}}

In \cite{ChaRei20:FlagMethodAnyD} the authors proposed two flag error-correction schemes for any \ac{CSS} code and any code distance $d$.
In the first scheme, per each controlled operation they use $d-1$ flags, resulting in a total of $n_\mathrm{flag}(\gamma_i, d) = \gamma_i(d-1)$.
In the second scheme, they propose an optimization of the first scheme in which the number of flags per controlled operation follow the pattern $0, 1, 2, \dots, d-2, d-1, d-1, \dots, d-1, d-2, \dots, 2, 1, 0$, instead of the constant pattern $d-1, \dots, d-1$.
This results in 
\begin{align}
    &n_\mathrm{flag}(\gamma_i, d) \notag \\ 
    &= 
    \begin{dcases}
        \binom{\lceil\frac{\gamma_i}{2}\rceil}{2} + \binom{\lfloor\frac{\gamma_i}{2}\rfloor}{2} &\left\lfloor\frac{\gamma_i}{2}\right\rfloor < d\\
        2\binom{d}{2} + (d-1)(\gamma_i-2d) &\text{otherwise}
    \end{dcases}
\end{align}

\subsubsection{Flag-based Prabhu and Reichardt scheme~\cite{Pra2023:flagFT}} In \cite{Pra2023:flagFT} the authors propose flag error-correction schemes for distance $d = 5$ and $d = 7$ codes. 
Specifically, we have
\begin{align}
 n_\mathrm{flag}(\gamma_i, 5) 
    &= 
    \begin{dcases}
        6 & \gamma_i \leq 8, \gamma_i \,\, \text{even} \\
        \gamma_i / 2 + 1 &\gamma_i > 8, \gamma_i \,\, \text{even}
    \end{dcases}
\end{align}
and
\begin{align}
 n_\mathrm{flag}(\gamma_i, 7) = \gamma_i + 1
\end{align}

\subsection{Faulty Locations in Cat state-based Syndrome Extraction}

For a cat state-based syndrome extraction, the total number of faulty locations involved in measuring the $i$-th stabilizer generator is given by
\begin{align}
\label{eq:nFPiCat}
    n_{\mathrm{FL},i}(d) = 4\gamma_i + f_{\text{cat}}(\gamma_i, d)
\end{align}
where $\gamma_i$ is the weight of the $i$-th generator, and $f_{\text{cat}}(\gamma_i, d)$ denotes the number of faulty locations involved in generating a cat state with size $\gamma_i$, achieving fault tolerance up to distance $d$.
In \eqref{eq:nFPiCat}, the term $4\gamma_i$ accounts for the two-qubit gates, in addition to the final $\gamma_i$ Hadamard gates and measurements performed on each qubit of the cat state.

\subsubsection{Cat state preparation based on generator measurement}

It is possible to prepare a cat state by measuring its $\M{X}$ generator, i.e. a Pauli $\M{X}$ on every qubit of the final state. 
In order to make this state preparation \ac{FT}, it is sufficient to apply one of the previous flag schemes to the measurement of this operator.
In this case, we obtain
\begin{align}
\label{eq:fCAT}
    f_{\text{cat}}(\gamma_i, d) = 3\gamma_i + 4 + 6\, n_\mathrm{flag}(\gamma_i, d)\,,
\end{align}
where $3\gamma_i$ accounts for the two-qubit gates and initializations, and the term $4$ corresponds to the initialization, measurement, and Hadamard gates applied to the qubit performing the generator measurement to build the cat state.
The term $6\, n_\mathrm{flag}(\gamma_i, d)$ accounts for the six faulty locations (including both two-qubit gates and measurements) per flag qubit introduced to ensure fault tolerance in the cat state preparation.

\subsection{Derivation of $\gamma_i$, $v_j^{(\mathrm{z})}$, and $v_j^{(\mathrm{x})}$ for some quantum codes}

For all the following codes, we have that $|\mathcal{G}_x| = |\mathcal{G}_z| = (n-k)/2$.

\subsubsection{Surface code}
For a symmetric $[[2d^2-2d+1, 1, d]]$ surface code~\cite{DenKitLan:02, FowMarMar:12, Rof:19}, the weight of the $i$-th generator is\footnote{Hereafter, we enumerate the generators in order of increasing weight.} 
\begin{align}
    \gamma_i = \begin{dcases}
        3 & 1 \le i \le 4(d-1) \\
        4 & 4(d-1) < i \le 2d(d-1)
    \end{dcases}
\end{align}
Regarding $v_j^{(\mathrm{z})}$ we have that
\begin{align}
    v_j^{(\mathrm{z})} = \begin{dcases}
        1 & 1\le j \le 2d\\
        2 & 2d < j \le 2d^2-2d+1\\
    \end{dcases}
\end{align}
Since we are considering symmetric surface codes, $v_j^{(\mathrm{z})} = v_{\ell}^{(\mathrm{x})}$ for some permutation matrix with ones in position $(j, \ell)$.

\subsubsection{Rotated surface code}
For a symmetric $[[d^2, 1, d]]$ rotated surface code~\cite{HorFowDev:12, ForValChi24:JSAC}, the weight of the $i$-th generator is
\begin{align}
    \gamma_i = 
    \begin{dcases}
        2 & 1 \le i \le 4(d-2) \\
        4 & 4(d-2) < i \le d^2-1
    \end{dcases}
\end{align}
Regarding $v_j^{(\mathrm{z})}$ we have that
\begin{align}
    v_j^{(\mathrm{z})} = \begin{dcases}
        1 & 1\le j \le 2d\\
        2 & 2d < j \le d^2\\
    \end{dcases}
\end{align}
Since we are considering symmetric rotated surface codes, $v_j^{(\mathrm{z})} = v_{\ell}^{(\mathrm{x})}$ for some permutation matrix with ones in position $(j, \ell)$.

\subsubsection{M\"obius code}
For a symmetric $[[2d^2-d, 1, d]]$ M\"obius code~\cite{ValForChi25:CylMob}, the weight of the $i$-th generator is
\begin{align}
    \gamma_i = 
    \begin{dcases}
        3 & 1 \le i \le 2d \\
        4 & 2d < i \le 2d^2-d- 1\,.
    \end{dcases}
\end{align}
Regarding $v_j^{(\mathrm{z})}$ we have that
\begin{align}
    v_j^{(\mathrm{z})} = \begin{dcases}
        1 & 1\le j \le 2d\\
        2 & 2d < j \le 2d^2-d\\
    \end{dcases}
\end{align}
Since we are considering symmetric M\"obius codes, $v_j^{(\mathrm{z})} = v_{\ell}^{(\mathrm{x})}$ for some permutation matrix with ones in position $(j, \ell)$.

Note that the same $\gamma_i$, $v_j^{(\mathrm{z})}$, and $v_j^{(\mathrm{x})}$ also hold for the $[[2d^2 - d, 1, d]]$ cylindrical code, as the Möbius code is derived from the cylindrical code by introducing a non-trivial fiber bundle twist, effectively replacing the trivial bundle structure with a Möbius one~\cite{ValForChi25:CylMob}.

\subsubsection{Honeycomb color code}
For a $[[(3d^2+1)/4, 1, d]]$ honeycomb color code with odd $d$~\cite{Bom06:colorCodes}, the weight of the $i$-th generator is
\begin{align}
    \gamma_i = 
    \begin{dcases}
        4 & 1 \le i \le 3(d-1) \\
        6 & 3(d-1) < i \le 3(d^2 - 1)/4
    \end{dcases}
\end{align}
Regarding $v_j^{(\mathrm{z})}$ we have that
\begin{align}
    v_j^{(\mathrm{z})} = \begin{dcases}
        1 & 1 \le j \le 3\\
        2 & 3 < j \le 3(d-1)\\
        3 & 3(d-1) < j \le (3d^2+1)/4\\
    \end{dcases}
\end{align}
Since the honeycomb color code is self-dual, we have that $v_j^{(\mathrm{z})} = v_{j}^{(\mathrm{x})}$.

\subsubsection{Square-octagon color code}
For a $[[(d^2-1)/2 + d, 1, d]]$ square-octagon color code with odd $d$~\cite{Bom06:colorCodes}, the weight of the $i$-th generator is
\begin{align}
    \gamma_i = 
    \begin{dcases}
        4 & 1 \le i \le (d-1)(d+5)/4\\
        8 & (d-1)(d+5)/4 < i \le (d^2 - 1)/2 + d - 1
    \end{dcases}
\end{align}
Regarding $v_j^{(\mathrm{z})}$ we have that
\begin{align}
    v_j^{(\mathrm{z})} = \begin{dcases}
        1 & 1 \le j \le 3\\
        2 & 3 < j \le 3(d-1)\\
        3 & 3(d-1) < j \le (d^2-1)/2 + d\\
    \end{dcases}
\end{align}
Since the square-octagon color code is self-dual, we have that $v_j^{(\mathrm{z})} = v_{j}^{(\mathrm{x})}$.

\subsubsection{Gross code}
For the $[[144, 12, 12]]$ Gross code, which is a bivariate bicycle quantum low-density parity-check code, all stabilizer generators have the same weight, namely
\begin{align}
    \gamma_i = 6
\end{align}
for all $i$. Moreover, each physical qubit participates in three $\mathbf{Z}$-type generators and three $\mathbf{X}$-type generators, so that
\begin{align}
    v_j^{(\mathrm{z})} = v_j^{(\mathrm{x})} = 3
\end{align}
for all $j$.

\section{Numerical Results}
\label{sec:NumRes}

In this section, we present the derived performance bounds and validate them through Monte Carlo simulations. The simulator was implemented in C++ and include:
$i)$ the \ac{FT}-\ac{QEC} circuit for syndrome extraction, accounting for faulty gates and measurements;
$ii)$ the decoder for the adopted \ac{QEC} code; and
$iii)$ the decoder for the flag qubits.

The \ac{FT} syndrome extraction is performed iteratively until $t+1$ consecutive identical syndromes are observed. 
During each generator extraction step, the flag-based decoder is applied to prevent error propagation and to ensure that the \ac{FT} gadget condition is satisfied. 
Once the extraction process halts, the final syndrome is passed to the main \ac{QEC} decoder, which attempts to identify and correct the underlying error.
To assess decoder failure, we verify whether the decoder successfully corrects the errors present at the input of the final syndrome extraction circuit. 
If the applied correction also eliminates the error at the output of the last syndrome extraction, no residual error is recorded. 
Otherwise, a residual error is registered in the simulation.

In the following, we evaluate the performance of both the Steane code and the surface code. 
For surface code simulations, we employ the \ac{MWPM} decoder, while for the Steane code, a \ac{MW} decoder is used.
For \ac{FT} syndrome extraction, we adopt the optimized flag-based circuit shown in Fig.~\ref{fig:cat_flag}a, with the number of flag qubits $n_\mathrm{flag}$ described by \eqref{eq:nflagOpt}.
It is worth noting that both the Steane and surface codes require a \ac{FT} extraction of weight-four generators, as weight-three generators are intrinsically \ac{FT}.
Regarding the error model, a depolarizing channel with parameter $p_\mathrm{FT}$ is applied at every faulty location, while an initial depolarizing channel with parameter $p$ is applied to the codeword at the start of the simulation.

\begin{figure}[t]
	\centering
	\resizebox{\columnwidth}{!}{ 
%
%
\definecolor{mycolor1}{rgb}{0.92157,0.58039,0.52549}%
\definecolor{mycolor2}{rgb}{0.85000,0.32500,0.09800}%
\definecolor{mycolor3}{rgb}{0.92900,0.69400,0.12500}%
\definecolor{mycolor4}{rgb}{0.49400,0.18400,0.55600}%
\begin{tikzpicture}

\begin{axis}[%
name = plot,
width=4.5in,
height=3.5in,
at={(0in,0in)},
scale only axis,
xmode=log,
xmin=0.001,
xmax=0.2,
xminorticks=true,
xlabel style={font=\color{white!15!black}, font = \Large},
xlabel={Channel Error Probability $p$},
tick label style={black, semithick, font=\Large},
ymode=log,
ymin=1e-05,
ymax=1,
yminorticks=true,
ylabel style={font=\color{white!15!black}, font = \Large},
ylabel={$P_\mathrm{fail,dec}$},
axis background/.style={fill=white},
title style={font=\bfseries},
legend style={at={(0.03,0.97)}, anchor=north west, legend cell align=left, align=left, draw=white!15!black}
]
\addplot [color=goodBlue, line width=1.5pt, only marks, draw=none, mark=o, mark size = 3pt, mark options={solid, goodBlue}]
  table[row sep=crcr]{%
0.001	2.14717e-05\\
0.002	8.13187e-05\\
0.005	0.00045092\\
0.01	0.00206817\\
0.02	0.00939585\\
0.05	0.0397141\\
0.1	0.15748\\
0.2	0.393701\\
};
\addlegendentry{Simulation}

\addplot [color=goodRed, line width=1.5pt, only marks, draw=none, mark=square, mark size = 3pt, mark options={solid, goodRed}]
  table[row sep=crcr]{%
0.001	1.565e-05\\
0.002	5.88345e-05\\
0.005	0.000418605\\
0.01	0.00162807\\
0.02	0.00607595\\
0.05	0.034618\\
0.1	    0.112782\\
0.2	    0.307062\\
};
\addlegendentry{Simulation w/ Ideal Gates}

\addplot [color=black, dashed, line width=1.5pt]
  table[row sep=crcr]{%
0.001	5.32227918111916e-05\\
0.002	0.000211501249621771\\
0.005	0.00129619634033729\\
0.01	0.00501842171350464\\
0.02	0.0188126873619985\\
0.05	0.0970577801094188\\
0.1	0.28476809892489\\
0.2	0.637215331730078\\
0.4	0.946580983478709\\
};
\addlegendentry{Upper Bound \eqref{eq:Pfail1}}

\addplot [color=goodBlue, dashed, line width=1.5pt]
  table[row sep=crcr]{%
0.001	2.72303321844358e-05\\
0.002	0.000108633013989556\\
0.005	0.000673502239647017\\
0.01	0.00265705960436535\\
0.02	0.0103253299014846\\
0.05	0.0586491226995957\\
0.1	0.195714315516594\\
0.2	0.521356606920848\\
};
\addlegendentry{Upper Bound \eqref{eq:Pfail2}}

\addplot [color=goodRed, densely dotted, line width=1.5pt]
  table[row sep=crcr]{%
0.001	1.62867249626382e-05\\
0.002	6.49609321393961e-05\\
0.005	0.000402536341848863\\
0.01	0.00158724627832006\\
0.02	0.00616921461162673\\
0.05	0.0353530979166668\\
0.1	    0.1221382\\
0.2	0.362116266666666\\
0.4	0.7833088\\
};
\addlegendentry{Upper Bound \eqref{eq:PLBD} with \eqref{eq:PFGivenE_BDwithBeta} - $p_{FT} = 0$}

\end{axis}
\end{tikzpicture}%
	} 
	\caption{$[[7, 1, 3]]$ Steane code with $p_\mathrm{FT} = p/100$.}
    \label{fig:SteanePerf_p100}
\end{figure}
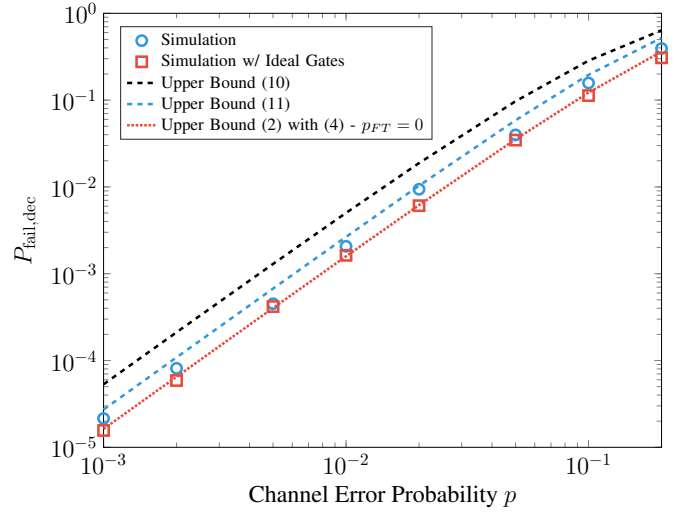

In Fig.~\ref{fig:SteanePerf_p100} and Fig.~\ref{fig:SurfacePerf_p10}, we present the simulation results and performance bounds for the $[[7,1,3]]$ Steane code with $p_\mathrm{FT} = p/100$ and the $[[13,1,3]]$ surface code with $p_\mathrm{FT} = p/10$. 
Specifically, we report results from simulations using both faulty and ideal quantum operations.
We also include the upper bound of the \ac{QEC} code, as given by~\eqref{eq:PLBD} with~\eqref{eq:PFGivenE_BDwithBeta}~\cite{ForValChi25:MacW,ForValChi24:JSAC}. 
For the \ac{FT}-\ac{QEC} performance limits, we provide both the simple bound in~\eqref{eq:Pfail1} and its enhanced version in~\eqref{eq:Pfail2}.
The results indicate that the refined bound~\eqref{eq:Pfail2} is tight for small values of $p_\mathrm{FT}$, though it becomes loose as $p_\mathrm{FT}$ increases. 
Since technological advancements are expected to reduce $p_\mathrm{FT}$, these bounds are anticipated to become tighter in realistic scenarios. 
Nevertheless, even when not tight, the upper bound remains useful in providing performance guarantees, ensuring that a given scheme can achieve at least a certain level of performance.
Despite its looser characterization, the bound in~\eqref{eq:Pfail1} is still valuable due to its dependence solely on the total number of faulty locations, $N_\mathrm{FL}$. 
For this reason, it is also reported.

\begin{figure}[t]
	\centering
	\resizebox{\columnwidth}{!}{ 
%
%
\definecolor{mycolor1}{rgb}{0.92157,0.58039,0.52549}%
\definecolor{mycolor2}{rgb}{0.85000,0.32500,0.09800}%
\definecolor{mycolor3}{rgb}{0.92900,0.69400,0.12500}%
\definecolor{mycolor4}{rgb}{0.49400,0.18400,0.55600}%
\begin{tikzpicture}

\begin{axis}[%
name = plot,
width=4.5in,
height=3.5in,
at={(0in,0in)},
scale only axis,
xmode=log,
xmin=0.001,
xmax=0.2,
xminorticks=true,
xlabel style={font=\color{white!15!black}, font = \Large},
xlabel={Channel Error Probability $p$},
tick label style={black, semithick, font=\Large},
ymode=log,
ymin=1e-05,
ymax=1,
yminorticks=true,
ylabel style={font=\color{white!15!black}, font = \Large},
ylabel={$P_\mathrm{fail,dec}$},
axis background/.style={fill=white},
title style={font=\bfseries},
legend style={at={(0.97,0.03)}, anchor=south east, legend cell align=left, align=left, draw=white!15!black}
]
\addplot [color=goodBlue, line width=1.5pt, only marks, draw=none, mark=o, mark size = 3pt, mark options={solid, goodBlue}]
  table[row sep=crcr]{%
0.001	8.14571e-05\\
0.002	0.000360112\\
0.005	0.00249417\\
0.01	0.00888573\\
0.02	0.0419375\\
0.05	0.227531\\
0.1	0.571429\\
0.2	0.813008\\
};
\addlegendentry{Simulation}

\addplot [color=goodRed, line width=1.5pt, only marks, draw=none, mark=square, mark size = 3pt, mark options={solid, goodRed}]
  table[row sep=crcr]{%
0.001	2.38469e-05\\
0.002	8.14299e-05\\
0.005	0.000482904\\
0.01	0.00206105\\
0.02	0.00841503\\
0.05	0.0484496\\
0.1	0.149365\\
0.2	0.372439\\
};
\addlegendentry{Simulation w/ Ideal Gates}

\addplot [color=black, dashed, line width=1.5pt]
  table[row sep=crcr]{%
0.001	0.00227452649992088\\
0.002	0.00866001273809702\\
0.005	0.0467998172949131\\
0.01	0.148197035541238\\
0.02	0.383651519495739\\
0.05	0.821600133321325\\
0.1	0.98069534766083\\
0.2	0.999808087590132\\
};
\addlegendentry{Upper Bound \eqref{eq:Pfail1}}

\addplot [color=goodBlue, dashed, line width=1.5pt]
  table[row sep=crcr]{%
0.001	0.000244144478034891\\
0.00142112270763801	0.000536320928245959\\
0.00201958975016438	0.00119865359950311\\
0.00287008485407157	0.00271938560667151\\
0.00407874275896902	0.00622693911156402\\
0.00579639395338497	0.0142501961406782\\
0.00823738706957101	0.0321251018165061\\
0.0117063378161711	0.069988851246986\\
0.0166361424938422	0.143898740089091\\
0.0236419998655007	0.271626599270101\\
0.0335981828628378	0.457441283621461\\
0.0477471406017529	0.671375725343787\\
0.0678545457339358	0.852609022271593\\
0.0964296357589577	0.95644443326499\\
0.137038345066317	0.992753731486066\\
0.194748303990876	0.999465167259568\\
0.276761237075423	0.999988189656288\\
0.393311678601869	0.999999961716939\\
0.558944157640337	0.999999999995713\\
0.794328234724281	1\\
};
\addlegendentry{Upper Bound \eqref{eq:Pfail2}}

\addplot [color=goodRed, densely dotted, line width=1.5pt]
  table[row sep=crcr]{%
0.001	1.87989692221665e-05\\
0.002	7.55029411616956e-05\\
0.005	0.000477327183478119\\
0.01	0.00194137850668971\\
0.02	0.007964320779869\\
0.05	0.0511276860644098\\
0.1	0.192628098379648\\
0.2	0.562669250052096\\
};
\addlegendentry{Upper Bound \eqref{eq:PLBD} with \eqref{eq:PFGivenE_BDwithBeta} - $p_{FT} = 0$}

\end{axis}
\end{tikzpicture}%
	} 
	\caption{$[[13, 1, 3]]$ Surface code with $p_\mathrm{FT} = p/10$.}
    \label{fig:SurfacePerf_p10}
\end{figure}
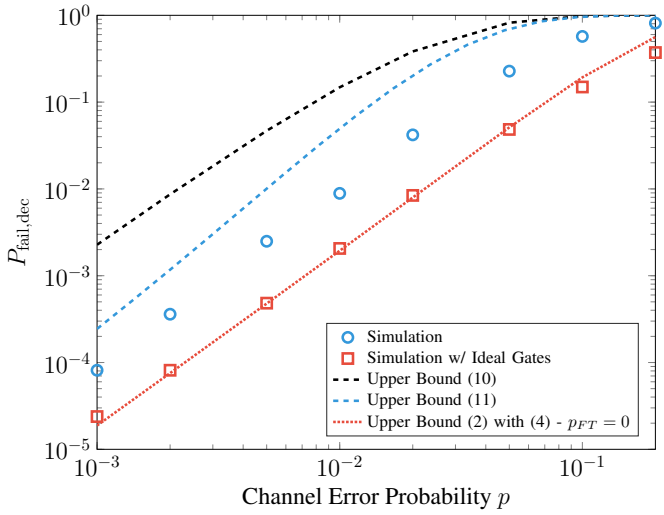

In Fig.~\ref{fig:SurfaceRes_p100}, we present simulations and analytical bounds to illustrate the impact of residual errors on overall performance. Specifically, we report the residual error probability $P_\mathrm{res}$, the decoding error probability $P_\mathrm{fail,dec}$, and the overall failure probability $P_\mathrm{fail}$.
The failure probability accounts for both undetectable and detectable errors. Undetectable errors arise from logical faults caused by decoding failures or residual errors. 
Detectable errors, on the other hand, originate from incorrect syndromes input to the decoder or from residual errors.
These detectable errors can, in principle, be identified and corrected by an ideal round of error correction.
However, in practice, attempting to detect and correct such errors may itself introduce new residual errors due to imperfections in the error correction process.
From the plot, we observe that the asymptotic upper bound~\eqref{eq:PresUB_Flag} on $P_\mathrm{res}$ is very tight within its range of applicability, namely for $p \ll 1$ and $p_\mathrm{FT} \ll 1$. 
In contrast, the lower bound on the residual error~\eqref{eq:PresLB} is less tight, though it still provides a reasonable estimate of the actual curve location.
Finally, by combining these bounds with the decoding failure bound~\eqref{eq:Pfail2}, we obtain an upper bound on the overall failure probability, as expressed in~\eqref{eq:Pfail_Def}.

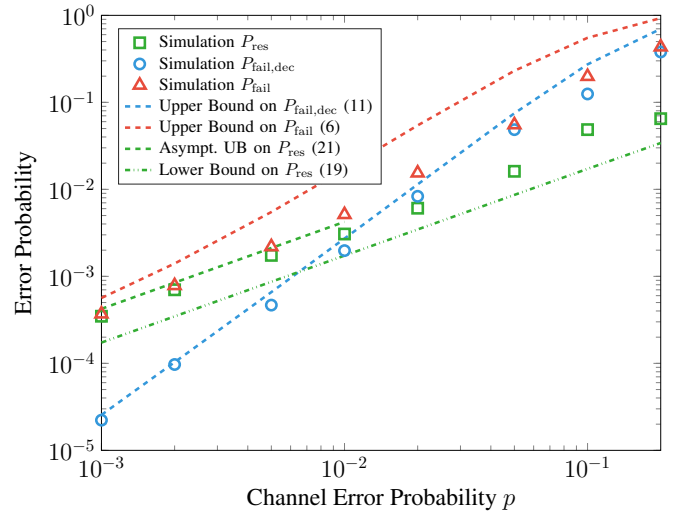
\begin{figure}[t]
	\centering
	\resizebox{\columnwidth}{!}{ 
%
%
\definecolor{mycolor1}{rgb}{0.92157,0.58039,0.52549}%
\definecolor{mycolor2}{rgb}{0.85000,0.32500,0.09800}%
\definecolor{mycolor3}{rgb}{0.92900,0.69400,0.12500}%
\definecolor{mycolor4}{rgb}{0.49400,0.18400,0.55600}%
\begin{tikzpicture}

\begin{axis}[%
name = plot,
width=4.5in,
height=3.5in,
at={(0in,0in)},
scale only axis,
xmode=log,
xmin=0.001,
xmax=0.2,
xminorticks=true,
xlabel style={font=\color{white!15!black}, font = \Large},
xlabel={Channel Error Probability $p$},
tick label style={black, semithick, font=\Large},
ymode=log,
ymin=1e-05,
ymax=1,
yminorticks=true,
ylabel style={font=\color{white!15!black}, font = \Large},
ylabel={Error Probability},
axis background/.style={fill=white},
title style={font=\bfseries},
legend style={at={(0.03,0.97)}, anchor=north west, legend cell align=left, align=left, draw=white!15!black}
]
\addplot [color=goodGreen, line width=1.5pt, draw=none, only marks, mark=square, mark size = 3pt, mark options={solid, goodGreen}]
  table[row sep=crcr]{%
0.001	0.000348064\\
0.002	0.000701024\\
0.005	0.00174134\\
0.01	0.00306723\\
0.02	0.00607038\\
0.05	0.016129\\
0.1	0.0485981\\
0.2	0.0649351\\
};
\addlegendentry{Simulation $P_\mathrm{res}$}

\addplot [color=goodBlue, line width=1.5pt, only marks, draw=none, mark=o, mark size = 3pt, mark options={solid, goodBlue}]
  table[row sep=crcr]{%
0.001	2.22375e-05\\
0.002	9.70205e-05\\
0.005	0.000467183\\
0.01	0.00197685\\
0.02	0.00831947\\
0.05	0.0485437\\
0.1	0.124844\\
0.2	0.379507\\
};
\addlegendentry{Simulation $P_\mathrm{fail, dec}$}

\addplot [color=goodRed, line width=1.5pt, only marks, draw=none, mark=triangle, mark size = 4pt, mark options={solid, goodRed}]
  table[row sep=crcr]{%
0.001	0.000369528226474211\\
0.002	0.000782285693639651\\
0.005	0.0021914747949672\\
0.01	0.0051290564469309\\
0.02	0.015409193620077\\
0.05	0.0552025512811\\
0.1	0.198424871212\\
0.2	0.4330710861949\\
};
\addlegendentry{Simulation $P_\mathrm{fail}$}

\addplot [color=goodBlue, dashed, line width=1.5pt]
  table[row sep=crcr]{%
0.001	2.55595264635478e-05\\
0.002	0.000103067204192953\\
0.005	0.00065880341764013\\
0.01	0.00272112552150339\\
0.02	0.0114120469966629\\
0.05	0.0748896539009635\\
0.1	0.273205674127728\\
0.2	0.698785352503042\\
};
\addlegendentry{Upper Bound on $P_\mathrm{fail,dec}$~\eqref{eq:Pfail2}}

\addplot [color=goodRed, dashed, line width=1.5pt]
  table[row sep=crcr]{%
0.001	0.00056410430286013\\
0.002	0.0014036655331373\\
0.005	0.00548648659915418\\
0.01	0.0170041123564088\\
0.02	0.05435999982844\\
0.05	0.230661191491344\\
0.1	0.551981213192931\\
0.2	0.929668227173945\\
};
\addlegendentry{Upper Bound on $P_\mathrm{fail}$~\eqref{eq:Pfail_Def}}

\addplot [color=goodGreen, dashed, line width=1.5pt]
  table[row sep=crcr]{%
0.001	0.000423333333333333\\
0.002	0.000846666666666667\\
0.005	0.00211666666666667\\
0.01	0.00423333333333333\\
};
\addlegendentry{Asympt. UB on $P_\mathrm{res}$ \eqref{eq:PresUB_Flag}}

\addplot [color=goodGreen, dashdotdotted, line width=1.5pt]
  table[row sep=crcr]{%
0.001	0.000173319178481203\\
0.002	0.000346610050071849\\
0.005	0.000866312865668099\\
0.01	0.00173191848080032\\
0.02	0.0034610100667621\\
0.05	0.00863136553856527\\
0.1	0.0171924786938582\\
0.2	0.0341060331148166\\
};
\addlegendentry{Lower Bound on $P_\mathrm{res}$ \eqref{eq:PresLB}}

\end{axis}
\end{tikzpicture}%
	} 
	\caption{$[[13, 1, 3]]$ Surface code with $p_\mathrm{FT} = p/100$.}
    \label{fig:SurfaceRes_p100}
\end{figure}


To summarize the source of errors we have investigated, residual errors are errors that remain undetected at the end of a given error-correction cycle but can be detected and mitigated by the next error correction cycle if they do not form a logical operator. They are therefore an inherent feature of QEC-based architectures, where each round may both correct previous errors and introduce new residual ones. 
However, these kind of errors could be sometimes mitigated using post-processing.
For instance, during the final measurement stage of a quantum computation, residual errors could be corrected acting on the measured codeword qubits.
In contrast, decoding failures arise when the decoder applies an incorrect correction, often due to an excessive number of errors, and result in non-recoverable logical errors. The total failure probability accounts for both contributions and corresponds to an overall performance metric (codeword error rate), which is the probability that the output state differs from the input.

The goal of this work was to propose a general framework for the analysis of fault-tolerant Shor-style syndrome extraction circuits, deliberately avoiding commitment to specific protocols in order to remain as broadly applicable as possible.
Nevertheless, a number of recent works illustrate important directions in the design of syndrome extraction schemes.
For example, \cite{LowDepthSyndExt23:BhaSte} develops highly optimized Steane-code circuits using ZX-based rewrites together with a dynamic fallback mechanism; \cite{ParallelSyndExt23:PeiChi} introduces parallel syndrome extraction for distance-three CSS codes using shared flag qubits to reduce circuit depth; and \cite{UltraLowSyndExt25:PooBolKis} proposes adaptive flag-based protocols where the choice of stabilizer measurements depends on past syndromes. 
The first approach is tailored to Steane-style error correction and therefore does not directly fall within the scope of our framework. 
In principle, however, the framework could be adapted to accommodate the latter two approaches, taking into account that they rely on predefined, though potentially adaptive, sequences of stabilizer measurements.

\section{Conclusions}\label{sec:conclusions}
This work presents an analytical framework for evaluating the performance of Shor-style \ac{FT}-\ac{QEC} under circuit-level noise. 
By distinguishing between decoding failures and residual errors, we provide a clearer understanding of the distinct error mechanisms that affect \ac{FT}-\ac{QEC} schemes. 
Unlike prior approaches that often rely on specific circuit or decoder details, our analysis remains applicable and captures performance limitations based solely on structural parameters such as the number of flag qubits and quantum operations.

Our findings underscore the intrinsic vulnerability of \ac{FT}-\ac{QEC} to residual errors, an issue often underrepresented in conventional analyses. 
These residual errors impose fundamental limits on the achievable reliability of quantum error correction, even when the code itself is theoretically sound.
The analytical bounds we derive are further substantiated by Monte Carlo simulations, which closely track the behavior of actual \ac{FT}-\ac{QEC} processes under realistic fault models. Together, these results provide a valuable tool for assessing and comparing \ac{FT}-\ac{QEC} designs at an early stage, before committing to specific circuit architectures or decoding algorithms.

Future work may extend this framework to incorporate circuit-specific optimizations or analyze the trade-offs introduced by different flag decoding strategies, further bridging the gap between theoretical \ac{QEC} models and practical quantum hardware.


\appendices

\section{Proof of \eqref{eq:PSofs}}
\label{app:PSofs}

Let $\rv{L}$ be the random variable representing the number of syndrome extraction attempts that encounter at least one error before observing the first sequence of $t+1$ consecutive error-free extractions.
Each syndrome extraction contains $N_\mathrm{FL}$ faulty locations. The probability that a single extraction is error-free is given by $(1 - p_\mathrm{FT})^{N_\mathrm{FL}}$.
As shown in \cite{Aki94:kConsecutiveSucc}, the probability mass function of $\rv{L}$ follows a geometric-like distribution. 
In our case, it is described by~\eqref{eq:Pellt}.

We compute the conditional distribution $\Prob{S = s \mid L = \ell}$, where $S$ denotes the total number of faults and $L$ the number of failed extractions. 
Conditioned on $L=\ell$, there are $\ell N_{\mathrm{FL}}$ faulty locations across the $\ell$ failed rounds. 
Let $U$ be the set of all fault patterns with exactly $s$ faulty locations, so that $|U| = \binom{\ell N_{\mathrm{FL}}}{s}$. 
For each round $r \in \{1,\dots,\ell\}$, let $A_r \subseteq U$ be the subset of patterns in which round $r$ contains no fault. We are interested in patterns in which every round contains at least one fault, i.e., $U \setminus \bigcup_{r=1}^{\ell} A_r$. By inclusion-exclusion principle~\cite{Feller91:Probability}, we have
\begin{align}
\left|U \setminus \bigcup_{r=1}^{\ell} A_r \right|
=
\sum_{j=0}^{\ell} (-1)^j \binom{\ell}{j} \binom{(\ell-j)N_{\mathrm{FL}}}{s}
\end{align}
since choosing $j$ empty rounds leaves $(\ell-j)N_{\mathrm{FL}}$ available locations for the $s$ faults. Each weight-$s$ pattern occurs with probability $p_{\mathrm{FT}}^{\,s}(1-p_{\mathrm{FT}})^{\ell N_{\mathrm{FL}}-s}$, and the conditioning event that every round contains at least one fault has probability $\bigl[1-(1-p_{\mathrm{FT}})^{N_{\mathrm{FL}}}\bigr]^\ell$. Therefore,
\begin{align}
&\Prob{S=s \mid L=\ell} \notag \\ 
&=
\frac{
p_{\mathrm{FT}}^{\,s}(1-p_{\mathrm{FT}})^{\ell N_{\mathrm{FL}}-s}
\sum_{j=0}^{\ell} (-1)^j \binom{\ell}{j} \binom{(\ell-j)N_{\mathrm{FL}}}{s}
}{
\bigl[1-(1-p_{\mathrm{FT}})^{N_{\mathrm{FL}}}\bigr]^\ell
}\,.
\end{align}
Finally, the marginal distribution is obtained via the law of total probability, 
$\Prob{S = s}
=
\sum_{\ell}
\Prob{S = s \mid L = \ell} \Prob{L = \ell}$, validating~\eqref{eq:PSofs}.

\section{$N$ failures before $k$ consecutive successes in bernoulli trials}
\label{app:failsBefkSucc}
Let $X_1, X_2, \dots$ be an \ac{i.i.d.} Bernoulli sequence with $\mathbb{P}(X_i = 1) = p$ and  $\mathbb{P}(X_i = 0) = q = 1-p$, where $1$ denotes success and $0$ denotes failure.
Define the random variable $\rv{N}$ as the number of failures observed before the first occurrence of $k$ consecutive successes.
Then, $\rv{N}$ follows a geometric distribution with parameter $p^k$ 
\begin{align}
    \Prob{\rv{N} = n} = p^k (1 - p^k)^n\,.
\end{align}

To prove this, let us observe first that before the first run of $k$ consecutive successes occurs, any failure resets
the current count of consecutive successes to zero.
If exactly $n$ failures occur before the first run of $k$ successes, the sequence must have the structure
\begin{align}
    S^{a_0} F \, S^{a_1} F \cdots F \, S^{a_{n-1}} F \, S^k \notag
\end{align}
where $S$ denotes a success, $F$ denotes a failure, and $a_i \in \{0,1,\dots,k-1\}$.
This condition is necessary because before each failure we may have any number
of consecutive successes strictly less than $k$, otherwise the run of $k$
successes would already have occurred.
For a fixed choice of $(a_0,\dots,a_{n-1})$, the probability of the sequence is
\begin{align}
    p^{a_0 + \cdots + a_{n-1} + k} q^n \,. \notag
\end{align}
Summing over all $a_i\in\{0,\dots,k-1\}$ and using the geometric series, we obtain
\begin{align}
    \Prob{\rv{N} = n} = q^n p^k \left( \sum_{a=0}^{k-1} p^a \right)^n = p^k (1-p^k)^n\,.
\end{align}
Another proof of this, using probability generating functions, can be found in~\cite{Aki94:kConsecutiveSucc}.

\section{Proof of Theorem~\ref{th:Pfail2}}
\label{app:Pfail2}

Let $\mathcal{D}_f$ denote the decoding failure event. 
By conditioning on the number of errors $s$ that occur during the syndrome extraction process, represented by the random variable $\rv{S}$, we can split the problem into three cases: $i)$ $\rv{S}\le t$; $ii)$ $\rv{S}= t+1$; $iii)$ $\rv{S}>t+1$.
We observe that for both $ii)$ and $iii)$ we cannot use the \ac{FT} gadget condition.
However, if $\max_i \{\gamma_i/2\} \le t+1$, then the \ac{FT} gadget condition can be used for $ii)$ since the maximum number of errors that can be propagated to the data qubits are $\max_i \{\lfloor\gamma_i/2\rfloor\}$.
To account for this condition we define
\begin{align}
    \tau = \begin{dcases}
        t+1 & \max_i \{\gamma_i/2\} \le t+1 \\
        t & \text{otherwise}\,.        
    \end{dcases}
\end{align}
Then, we upper bound the decoding failure probability $P_\mathrm{fail, dec}$ as follows
\begin{align}
    P_\mathrm{fail, dec} &= \sum_s \Prob{\mathcal{D}_f \mid \rv{S} = s} \, p_{\rv{S}}(s) \notag \\
    &\le \sum_{s > \tau} p_{\rv{S}}(s) + \sum_{s \le \tau} \Prob{\mathcal{D}_f \mid \rv{S} = s} \, p_{\rv{S}}(s)\,.
\end{align}

When $\rv{S} \le \tau$, there are two potential sources of decoder failure: one arises from too many errors on the data qubits, and the other from repeated errors in the syndrome measurements. 
The latter, a failure due to measurement errors, can only occur when $\rv{S} = t+1$. 
To account for this, we include a correction term of the form $(\tau - t) p_{\rv{S}}(t+1) P_\mathrm{meas}$, where $(\tau - t)$ acts as an indicator (equal to 1 if $t < \tau$ and 0 otherwise), and $P_\mathrm{meas}$ is an upper bound on the probability that $t+1$ consecutive measurement errors occur. 
The derivation of $P_\mathrm{meas}$ is deferred to later in this appendix. 
Note that the joint occurrence of both error types is ignored in this upper bound, as it does not affect the validity of the result and does not significantly impact on it.

We now refine the term that accounts for decoding failures caused by an excessive number of errors on the data qubits. 
By invoking the \ac{FT} gadget condition, we assume that out of $\rv{S}$ total faults, $\rv{Q}$ are effectively propagated to the data qubits. 
Let $P_\mathrm{fd}$ denote the probability that an error at a faulty location leads to an error on a data qubit. 
Then, the conditional probability distribution of $\rv{Q}$ given $\rv{S} = s$ follows a binomial form
\begin{align}
p_{\rv{Q}|\rv{S}}(q|s) = \binom{s}{q} P_\mathrm{fd}^q (1 - P_\mathrm{fd})^{s - q} \,.
\end{align}
The derivation of $P_\mathrm{fd}$ is deferred to later in this appendix.
Next, let us condition on the number of channel errors $\rv{C}$ and the total number of data qubit errors after syndrome extraction, denoted by the random variable $\rv{E}$. 
Recalling that $f(e)$ is the probability that the decoder fails the decoding when there are exactly $e$ errors in the data qubits, we have
\begin{align}
    &P_\mathrm{fail, dec} \le 1 - \sum_{s \le \tau} p_{\rv{S}}(s) + (\tau - t) p_{\rv{S}}(t+1) P_\mathrm{meas} \notag \\
    &+ \sum_c \sum_{s \le \tau} \sum_{q} \sum_e f(e) \, p_{\rv{E} \mid \rv{C}, \rv{Q}}(e \mid c, q) \, p_{\rv{Q} \mid \rv{S}}(q \mid s) \, p_{\rv{C}}(c) \, p_{\rv{S}}(s) \,. \notag
\end{align}

To simplify the analysis, we approximate $p_{\rv{E} \mid \rv{C}, \rv{Q}}(e \mid c, q)$ by a probability mass function derived from an idealized occupancy problem, which we denote by $g(e \mid c, q)$. 
This approximation ignores the possibility that multiple errors on the same qubit might cancel.
However, since our goal is to upper bound $P_\mathrm{fail, dec}$, this simplification is conservative and valid due to the fact we are considering more errors than the actual ones. 
Substituting this bound, we obtain
\begin{align}
    &P_\mathrm{fail, dec} \le 1 - \sum_{s \le \tau} p_{\rv{S}}(s) + (\tau - t) \, p_{\rv{S}}(t+1) \, P_\mathrm{meas} \notag \\
    &+ \sum_c \sum_{s \le \tau} \sum_{q} \sum_e f(e) \, g(e \mid c, q) \, p_{\rv{Q} \mid \rv{S}}(q \mid s) \, p_{\rv{C}}(c) \, p_{\rv{S}}(s) \,. \notag
\end{align}

We now derive all terms that were previously postponed.
Let us start with $g(e \mid c, q)$. 
Consider an occupancy problem with $n$ bins (representing data qubits) in which we firstly place $c$ channel errors (``channel balls'') into $c$ distinct bins. 
Next, we place $q$ syndrome extraction errors (``syndrome balls'') into the bins, allowing arbitrary overlaps.
Note that the minimum number of data qubit errors is $e = c$, which occurs when all syndrome balls fall into bins already occupied by channel balls. 
The maximum is $e = \min\{n, c+q\}$, when all syndrome balls land into unoccupied bins or when all bins end up occupied.
Let $\rv{U}$ denote the number of syndrome balls that land in bins not already occupied by channel balls. The distribution of $\rv{U}$ is binomial
\begin{align}
    p_{\rv{U}}(u) = \binom{q}{u} \left(1 - \frac{c}{n} \right)^u \left( \frac{c}{n} \right)^{q - u}\,.
\end{align}
Using the law of total probability, we write
\begin{align}
    g(e \mid c, q) = \sum_u g(e \mid c, q, u) \, p_{\rv{U}}(u)\,.
\end{align}
To compute $g(e \mid c, q, u)$, we determine the probability that exactly $e - c$ out of the $n - c$ bins unoccupied by channel balls are hit by the $u$ syndrome balls.
This is a classical occupancy problem, and by the inclusion-exclusion principle~\cite{Feller91:Probability}, we have
\begin{align}
    g(e \mid c, q, u) 
    = \binom{n - c}{e - c} \sum_{v = 0}^{e - c} (-1)^v \binom{e - c}{v} \left( \frac{e - c - v}{n - c} \right)^u\,.
\end{align}
This completes the derivation. The final expression for $g(e \mid c, q)$ follows by substitution and simplification using the binomial theorem.

\begin{figure}[t]
    \centering
    \resizebox{\columnwidth}{!}{ 
        \input{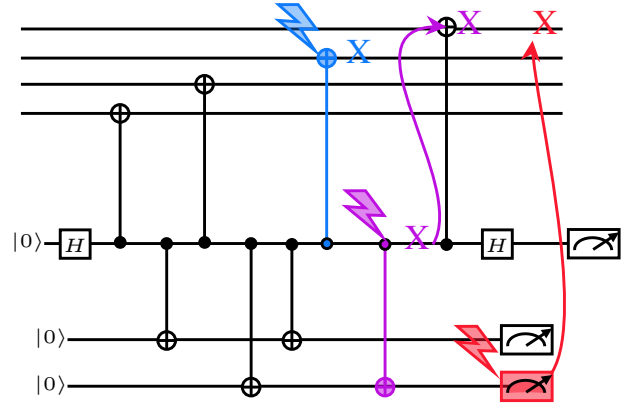}
    }
    \caption{Examples of syndrome extraction faults propagated to data qubits.
    (blue) An $\M{X}$ error occurring on a data qubit as a result of a fault during a CNOT gate with the syndrome qubit.
    (purple) An $\M{X}$ error occurring on the syndrome qubit, which is then propagated to a data qubit by a subsequent CNOT gate.
    (red) A flag measurement error leading to an incorrect correction, i.e., an $\M{X}$ operator is erroneously applied to a data qubit. }
    \label{fig:error_flag}
\end{figure}

Next, let us proceed in deriving $P_\mathrm{fd}$.
For clarity in the derivation, we illustrate in Fig.~\ref{fig:error_flag} several types of error events that will be counted in the following analysis.
To this aim, we estimate the number of faulty locations that can propagate errors to data qubits.  
Since we are aiming to derive an upper bound, we allow for possible overcounting in this analysis.
This is done for the sake of generality of the bound, not sticking with a particular \ac{FT} error correction logic using flags.
Firstly, note that in a flag-based syndrome extraction circuit, each generator has $\gamma_i$ faulty locations, specifically, those corresponding to the controlled gates acting on data qubits, that can propagate any of the three Pauli errors to the data.  
Thus, summing over all generators, we obtain a contribution of $\sum_i 3\gamma_i$ possible errors that affect the data qubits.
Secondly, consider the faulty location on the ancilla qubit. 
All faulty locations between the Hadamard gates (excluding the last controlled gate) can propagate errors to the data qubits through the controlled gates. 
For each such location, two out of the three Pauli errors can lead to propagation.  
Therefore, across all generators, this contributes a factor of $2\sum_i \gamma_i - 1 + 2 \, n_\mathrm{flag}(\gamma_i, d)$.
Note that, errors on the first Hadamard are not considered since they form an element of the stabilizer, having no impact on the codeword.
Lastly, consider the faulty locations on the flag qubits. 
Each of the four faulty locations associated with a flag ancilla could, in principle, propagate an error to the data qubits due to the \ac{FT} mechanism that uses flags to apply corrections. Again, two out of the three Pauli errors can lead to such propagation.  
Thus, summing over all generators, this yields a contribution of $2\sum_i 4 \, n_\mathrm{flag}(\gamma_i, d)$.
Collecting all these terms and dividing by the total number of possible error patterns $3N_\mathrm{FL}$ we find $P_\mathrm{fd}$ representing an upper bound to the probability that an error occurring on a faulty location is propagated to a data qubit.
As an example of overcounting, consider a case where, in a particular \ac{FT} syndrome extraction circuit, a single error propagates to one data qubit and is detected by the flags. 
The \ac{FT} error correction logic, depending on the flag outcome, could either apply no correction or correct the error, consistent with the \ac{FT} gadget condition.
In cases where the logic successfully corrects such an error, some of the faults counted in the analysis would not actually lead to errors in the data, resulting in a conservative overestimate.

Finally, let us derive $P_\mathrm{meas}$.
To this aim, we must upper bound the probability that $t+1$ consecutive measurement errors occur on the same syndrome bit.
Let us first focus on the $i$-th syndrome bit. 
As we did for $P_\mathrm{fd}$, we count the number of faults that could flip the measurement outcome.
Note that this also constitutes an overestimate, since some of these faults may propagate to the data qubits, thereby altering subsequent measurement outcomes.
Counting these error events, we note that at each faulty location on the ancilla qubit, two out of the three possible Pauli errors result in a flipped measurement outcome.  
For the flag qubits, only the faulty location after the first CNOT gate can propagate two out of three Pauli errors to the ancilla qubit.  
This results in a total of $2\gamma_i + 8 + 6\,n_\mathrm{flag}(\gamma_i, d)$ errors that can flip the measurement outcome of the $i$-th syndrome bit.  
Dividing by the total number of error patterns in the faulty locations, $3N_\mathrm{FL}$, we obtain the probability that the $i$-th syndrome bit is flipped, given that a single fault has occurred during this round of syndrome extraction.
Consequently, the probability that $t+1$ measurement errors affect the same syndrome bit (each occurring in a different round) is upper bounded by
\begin{align}
P_\mathrm{meas} = \sum_{i=1}^{n-k} \left( \frac{2\gamma_i/3 + 8/3 + 2\,n_\mathrm{flag}(\gamma_i, d)}{N_\mathrm{FL}} \right)^{t+1}\,.
\end{align}
We conservatively upper bound the probability that $t+1$ errors occur in consecutive rounds of syndrome extraction with $P_\mathrm{meas}$.

\section{Proof of Theorem~\ref{th:PresLB}}
\label{app:PresLB}

We consider as error sources only the faulty locations that act directly on the data qubits. 
Focusing on a single qubit, we separate the last faulty location from the earlier ones. 
A residual error can occur in two distinct cases: either all previous faulty locations introduce no error and the last one does, or the previous faulty locations introduce an error and the last one does not cancel it out.

Let us assume firstly that all $\mathbf{X}$-type stabilizer generators are measured, followed by all $\mathbf{Z}$-type generators. 
For a residual error to persist, the syndrome must remain unchanged during the final round. 
This happens when CZ gates used for measuring the $\mathbf{Z}$ generators introduce $\mathbf{Z}$ errors, since $\mathbf{X}$ or $\mathbf{Y}$ errors would alter the syndrome and be detected. 
These $\mathbf{Z}$ errors have to be accounted for all CZ gates, excluding the last one (i.e., $v_j^{(\mathrm{z})}-1$ when examining the qubit $j$), as well as to the final CX gate acting on the qubit under examination.
On the other hand, for the last CZ gate, we consider all possible errors.
It is important to note that other mechanisms can also lead to residual errors, though they are more difficult to systematically account for.
However, since we are deriving a lower bound, we can safely neglect these additional contributions, at the cost of obtaining a looser estimate.
For instance, $\mathbf{X}$ errors may also become residual if accompanied by measurement errors that render them undetectable.

For this reason, we obtain that a single qubit could have a residual error with probability greater than
\begin{align}
    q_j = r_{\mathrm{Z},j} \, (1-p_\mathrm{FT}/3) + (1-r_{\mathrm{Z},j}) \, p_\mathrm{FT}
\end{align}
where $r_{\mathrm{Z},j}$ represents the probability to have a residual error $\mathbf{Z}$ before entering on the last CZ, which is the probability to have an odd number of $\mathbf{Z}$ errors on the previous CZs and the last CX.
This probability problem is the same occurring on \ac{LDPC} codes when evaluating the probability to have an odd number of ones in a string of bits~\cite{Ryan:ChannelCodes}. 
In this way, we obtain
\begin{align}
    r_{\mathrm{Z},j} = \frac{1}{2} - \frac{1}{2} \left(1-\frac{2p_\mathrm{FT}}{3-2p_\mathrm{FT}}\right)^{v_j^{(\mathrm{z})}}
\end{align}
where the probability to have a single $\mathbf{Z}$ on a gate must be conditioned on the fact that both $\mathbf{X}$ and $\mathbf{Y}$ are not occurred.
To have a lower bound we can use for all qubits an error probability $q(p_\mathrm{FT})=\min_{j=1, \dots, n}\{ q_j\}$, leading to
\begin{align}
    q(p_\mathrm{FT}) = \frac{1}{2} + \frac{p_\mathrm{FT}}{3} - \frac{(3-4p_\mathrm{FT})^{v_\mathrm{m}+1}}{6\,(3-2p_\mathrm{FT})^{v_\mathrm{m}}}\,.
\end{align}
where $v_\mathrm{m} = \min_{j=1, \dots, n}\{v_j^{(\mathrm{z})}\}$.
On the other hand, if we consider the best possible ordering we have that $v_\mathrm{m} = 1$ since we have to account only for the second-to-last controlled gate.
Finally, by considering the intrinsic error correction of a stabilizer codes as in \eqref{eq:PLwithAz}, we obtain
\begin{align}
    P_\mathrm{res} \ge 1 - \sum_{w=0}^n \frac{A_w}{4^k} \, q^w(p_\mathrm{FT}) \, (1-q(p_\mathrm{FT}))^{n-w}
\end{align}
which concludes the proof.

\section{Proof of Theorem~\ref{th:PresUB_Flag}}
\label{app:PresUB_Flag}
To derive an upper bound for $P_\mathrm{res}$ we assume to be in the low-error regime ($p_\mathrm{FT} \ll 1$) and that if at least one faulty location, that in principle could lead to a residual error has run into an error, we fail.
Under the assumption that $p_\mathrm{FT} \ll 1$, we can restrict our analysis to single-error patterns.
In this case, it suffices to count how many faulty locations can lead to a residual error, and for each, determine which of the three Pauli errors ($\mathbf{X}$, $\mathbf{Y}$, or $\mathbf{Z}$) could cause it.

Let us start by considering the faulty locations on the data qubits.
We have $n$ faulty locations which are the last ones of each qubit.
We can account for them with a term $n\,p_\mathrm{FT}$ since all three Pauli errors could occur here.
Then, we have to consider the previous faulty locations due to other controlled-gates on the data qubits.
As a worst-case scenario, consider the ordering in which all $\mathbf{X}$-type generators are measured before any of the $\mathbf{Z}$-type generators.
Here, for the sake of clarity and presentation, we focus on \ac{CSS} codes where the $\mathbf{X}$ and $\mathbf{Z}$ generators are equal in number and have the same weights. 
In this way it does not matter if we measure first all $\mathbf{X}$-type generators or first all $\mathbf{Z}$-type generators.
If instead the $\mathbf{X}$ generators outnumber the $\mathbf{Z}$ generators, the same reasoning applies symmetrically by swapping $\mathbf{Z}$ with $\mathbf{X}$ and $v_j^{(\mathrm{z})}$ with $v_j^{(\mathrm{x})}$.
Under this assumption, we have that all CZs from $\mathbf{Z}$ generators could introduce residual $\mathbf{Z}$ errors, while both $\mathbf{X}$ and $\mathbf{Y}$ would be detected by the syndrome.
Moreover, the last CX on each qubit from the $\mathbf{X}$ generators could also have residual $\mathbf{Z}$ errors using the same argument. 
This ends up with a factor of $D_\mathrm{res}\,p_\mathrm{FT}/3$ where
$D_\mathrm{res} = \sum_{j=0}^{n} v_j^{(\mathrm{z})}$.
Accounting for a general \ac{CSS} code with asymmetries in the generators we have that $D_\mathrm{res} = \max\left\{ \sum_{j=0}^{n} v_j^{(\mathrm{z})}, \sum_{j=0}^{n} v_j^{(\mathrm{x})} \right\}$

We now consider the faulty locations on the ancillary qubit used for syndrome extraction. 
We focus only on $\mathbf{Z}$-type stabilizer generators, since any single-error pattern affecting an ancillary qubit used to measure $\mathbf{X}$-type generators is either detected by those generators themselves or by the subsequent measurement of the $\mathbf{Z}$-type generators.
To ensure that the syndrome remains unaffected, we discard both $\mathbf{Z}$ and $\mathbf{Y}$ errors that occur between the two Hadamard gates, as they would alter the measurement outcome. 
Considering the initialization step and the first Hadamard gate, we observe that all three Pauli errors can be neglected. 
Although some of these errors may not be detected by the syndrome, they effectively introduce elements of the stabilizer group and thus do not lead to residual errors.
Similarly, errors occurring after the second Hadamard gate, including those in the measurement process, can be ignored, since no subsequent two-qubit gates exist that could propagate such errors to the data qubits.
This leaves only the faulty locations between the two Hadamard gates as relevant. 
Each such location contributes a probability of $p_\mathrm{FT}$, and only $\mathbf{X}$ errors can lead to residual errors (as $\mathbf{Z}$ and $\mathbf{Y}$ are excluded). 
This results in a contribution of $A_\mathrm{res} \, p_\mathrm{FT}/3$, where $A_\mathrm{res} = \sum_{i\in\mathcal{G}_z} \gamma_i - 1 + 2n_\mathrm{flag}(\gamma_i, d)$, with $\mathcal{G}_z$ representing the set of the indices of all $\mathrm{Z}$ generators.
The ``$-1$'' term subtracts the last faulty location on the final controlled gate before the second Hadamard, which cannot propagate an error to the data and therefore does not contribute to residual errors.
Accounting for a general \ac{CSS} code with asymmetries in the generators we have that $A_\mathrm{res} = \sum_{i\in\mathcal{G}_\mathrm{M}} \gamma_i - 1 + 2n_\mathrm{flag}(\gamma_i, d)$, with $\mathcal{G}_\mathrm{M}$ representing the set with larger cardinality between $\mathcal{G}_\mathrm{z}$ and $\mathcal{G}_\mathrm{x}$.

Finally, we consider the faulty locations on the flag qubits.
Similar to the ancillary qubit, we focus only on $\mathbf{Z}$-type stabilizer generators, since any single-error pattern affecting an ancillary qubit used to measure $\mathbf{X}$-type generators is either detected by those generators themselves or by the subsequent measurement of the $\mathbf{Z}$-type generators.
Each flag qubit has four faulty locations.
We observe that both $\mathbf{Z}$ and $\mathbf{Y}$ errors occurring on the first two faulty locations would be propagated to the ancillary qubit, altering the syndrome. 
For this reason, only $\mathbf{X}$ errors are counted for them.
For the last two faulty location we need to count both $\mathbf{X}$ and $\mathbf{Y}$ errors.
Considering as a worst case that an erroneous flag always lead to a residual error we have to account for a term $F_\mathrm{res}\, p_\mathrm{FT}/6 + F_\mathrm{res}\, p_\mathrm{FT}/3$ where $F_\mathrm{res} = 4\sum_{i\in \mathcal{G}_z} n_\mathrm{flag}(\gamma_i, d)$.
Accounting for a general \ac{CSS} code with asymmetries in the generators we have that $F_\mathrm{res} = 4\sum_{i\in \mathcal{G}_\mathrm{M}} n_\mathrm{flag}(\gamma_i, d)$.

Gathering all the terms together, the claim follows for the flag-based syndrome extraction.
On the other hand, using a cat state, the reasoning presented above can be directly extended. 
When the cat state is constructed via its generator measurement, both equations~\eqref{eq:PresUB_Flag} and~\eqref{eq:PresUB_Flag_Param} remain applicable, with the sole modification that $A_\mathrm{res}$ must be increased by a factor of $\sum_{i \in \mathcal{G}_\mathrm{M}} \gamma_i$.
This is because the cat state is prepared using a flag-based circuit.


\bibliographystyle{IEEEtran}
\bibliography{Files/IEEEabrv,Files/StringDefinitions,Files/StringDefinitions2,Files/refs}

\end{document}